\documentclass[11pt,onecolumn,draftcls]{IEEEtran}
\usepackage{booktabs}
\usepackage{subfigure}
\usepackage{amsmath,epsfig,bm}
\newcommand{\y}{\mathbf{y}}

\newcommand{\x}{\mathbf{x}}

\renewcommand{\a}{\mathbf{a}}
\renewcommand{\b}{\mathbf{b}}   
\renewcommand{\r}{\mathbf{r}}
\newcommand{\g}{\mathbf{g}}
\newcommand{\f}{\mathbf{f}}

\renewcommand{\c}{\mathbf{c}}

\newcommand{\F}{\mathbf{F}}

\newcommand{\X}{\mathbf{X}}
\newcommand{\Y}{\mathbf{Y}}
\newcommand{\G}{\mathbf{G}}

\newcommand{\I}{\mathbf{I}}  
\newcommand{\W}{\mathbf{W}}
\newcommand{\V}{\mathbf{V}}

\newcommand{\U}{\mathbf{U}}
\newcommand{\Z}{\mathbf{Z}}

\newcommand{\R}{\mathbf{R}}

\renewcommand{\H}{\mathbf{H}}
\newcommand{\exE}{\mathbb{E}}
\newcommand{\ZV}{\bm{0}} 
\newcommand{\GAMMA}{\bm{\Gamma}} 
\newcommand{\Hnull}{\mathcal{H}_0}
\newcommand{\Halt}{\mathcal{H}_1}

\newcommand{\Real}{\mathcal{R}}
\newcommand{\Imag}{\mathcal{I}}

\usepackage{pifont}
\usepackage{graphicx}
\usepackage{amssymb}
\usepackage{amsmath}
\usepackage{cancel}
\usepackage[normalem]{ulem}
\usepackage[dvips]{color}
\usepackage{algorithm}  
\newfont{\fsc}{eusm10}                         

\DeclareMathOperator*{\argmax}{arg\,max}

\usepackage{lineno}
\modulolinenumbers[1]

\usepackage{algorithmic}
\usepackage{algorithm}

\newtheorem{theorem}{\noindent \textbf{Theorem}}
\newtheorem{lemma}{\noindent \textbf{Lemma}}

\title{Cooperative Spectrum Sensing for Amplify-and-Forward Cognitive Networks}
\author{
\IEEEauthorblockN{Ido Nevat$^1$,
                  Gareth W. Peters $^{2,3}$,
                  Jinhong Yuan$^4$
                  and Iain B. Collings $^1$\\
\IEEEauthorblockA{$^1$
Wireless \& Networking Tech. Lab, CSIRO, Sydney, Australia.}\\
\IEEEauthorblockA{$^2$
School of Mathematics and Statistics, University of NSW, Sydney, Australia.}\\
\IEEEauthorblockA{$^3$
CSIRO Mathematical and Information Sciences, Sydney, Australia.}\\
\IEEEauthorblockA{$^4$
School of Electrical Engineering,University of NSW, Sydney, Australia.}} }

\begin{document}
\maketitle
\begin{abstract}
We develop a framework for spectrum sensing in cooperative amplify-and-forward cognitive radio networks. We consider a stochastic model where relays are assigned in cognitive radio networks to transmit the primary user's signal to a cognitive Secondary Base Station (SBS).
We develop the Bayesian optimal decision rule under various scenarios of Channel State Information (CSI) varying from perfect to imperfect CSI.
In order to obtain the optimal decision rule based on a Likelihood Ratio Test (LRT), the marginal likelihood under each hypothesis relating to presence or absence of transmission needs to be evaluated pointwise. 
However, in some cases the evaluation of the LRT can not be performed analytically due to the intractability of the multi-dimensional integrals involved. In other cases, the distribution of the test statistic can not be obtained exactly.
To circumvent these difficulties we design two algorithms to approximate the marginal likelihood, and obtain the decision rule. The first is based on Gaussian Approximation where we quantify the accuracy of the approximation via a multivariate version of the Berry-Esseen bound. The second algorithm is based on Laplace approximation for the marginal likelihood, which results in a non-convex optimisation problem which is solved efficiently via Bayesian Expectation-Maximisation method.
We also utilise a Laguerre series expansion to approximate the distribution of the test statistic in cases where its distribution can not be derived exactly.
Performance is evaluated via analytic bounds and compared to numerical simulations.
\end{abstract}
\section{Introduction}
In recent years, cognitive radio \cite{haykin2005cognitive}, \cite{fette2006cognitive} has attracted intensive research due to the pressing demand for efficient frequency spectrum usage. In a cognitive radio system, the secondary users (SU) try to find ``blank spaces'', in which the licensed
frequency band is not being used by the Primary Base Station (PBS). A key requirement in cognitive radio is that the SUs need to vacate the frequency band as quickly as possible if the corresponding Primary User (PU) emerges. 

Spectrum sensing is a mandatory functionality in any CR-based wireless system that shares spectrum bands with primary services, such as the IEEE 802.22 standard \cite{IEEE802_22}. This standard proposes to reuse vacant spectrum in the TV broadcast bands.
There has been a significant amount of research on spectrum sensing for cognitive radio, see \cite{ghasemi2008spectrum}, \cite{akyildiz2008survey} for overviews.

Essentially, spectrum sensing can be cast as a decision making or classification problem. The secondary network needs to make a decision between the two possible hypotheses given an observation vector: that the frequency band is either occupied or vacant.
The more knowledge we have on the nature of the primary user's signal, the more reliable our decision. If no knowledge is assumed regarding the primary user, energy detector based approaches (also called radiometry) are the most common way of spectrum sensing
because of their low computational complexity. 
Cooperative networks can improve the performance of the network by enabling users to share information and create diversity. This helps to combat the detrimental effect caused by the fading channels.
In this context, cooperative spectrum sensing has been studied extensively as a promising alternative to improve the sensing performance. In \cite{quan2008optimal}, the authors proposed algorithms to optimise detection performance by operating over a linear combination of local test statistics from individual secondary users. In \cite{zhang2008cooperative}, the performance of cooperative spectrum sensing was derived. It was found that the optimal decision fusion rule to minimize the total error probability is the half-voting rule. In \cite{unnikrishnan2008cooperative}, centralized and decentralized detection schemes were developed.

In contrast to those methods, our system model for cooperative spectrum sensing contains the practical scenario of channel uncertainty.  This includes the case of partial CSI knowledge at the SBS or the more severe case of blind spectrum sensing.
We also assume that the relays have no processing capability, therefore are not capable of performing any local decisions. 
This is a practical scenario encountered in many relay networks \cite{chen2011}, \cite{chen2011cooperative}.
In order to perform LRT, the SBS performs a hypothesis test to decide whether the PBS is transmitting or idle in a given frame.
As we show, the densities involved in making a decision under this framework are intractable, meaning they can not be evaluated point wise. This is due to the fact that they involve multi-variate integrals which can not be solved analytically. 

\textbf{Contribution:}
\begin{enumerate}
\item We propose a novel statistical model to address the problem of spectrum sensing with partial CSI. To the best of our knowledge, a cooperative spectrum sensing, where both the channels from the PBS to the relays and the channel between the relays to the SBS are only partially known has not been addressed previously. In most cooperative CR systems, the relays perform a local soft or hard decision and then report their summary statistics to the SBS \cite{yucek2009survey}, \cite{srinivasa2010soft}. In the system model we present all the statistical processing is performed at the SBS, thus removing the computational complexity from the relays and placing it at the SBS, enabling the use of standard relays systems already developed and in operation, making such an approach widely applicable . 
	\item We derive the probabilities of detection and mis-detection as well as the associated optimal tests under several different scenarios. Some bounds have exact closed-form expressions while others have closed-form approximations that we derive via Laguerre series expansion.
	\item For the most complicated case of imperfect CSI, we derive two low complexity algorithms to perform spectrum sensing:
\begin{enumerate}
	\item [i.] The first is based on Gaussian approximation via moment matching. This results in a simple closed-form test statistic, for the decision process. In addition we study the approximation error providing closed form expression for the bound via Berry-Esseen theorem. 
	\item [ii.] The second is based on the Laplace approximation of the marginal likelihood which involves solving a non-convex optimisation problem.
\end{enumerate} 
\end{enumerate}

The paper is structured as follows: in Section \ref{System_Model} the stochastic system model is developed and the Bayesian estimation problem is presented. Section \ref{PERFECT_CSI} presents an analysis of the case of perfect CSI. In Section \ref{Laguerre_section} we develop the optimal decision rule and approximate the performance for the case of imperfect-perfect CSI case. In Section \ref{App_Algorithms} we present two novel algorithms to perform the hypothesis test in the case of imperfect-imperfect CSI. Section \ref{SimulationResults} presents extensive simulation results. Conclusions are provided in Section \ref{Conclusions}.

\textbf{Notation.}The following notation is used throughout: random variables are denoted by upper case letters and their realizations by lower case letters, and bold case will be used to denote a vector or matrix quantity.

\section{Problem definition and System model} \label{System_Model}
In general, in cooperative spectrum sensing model, the task of the SBS is to discriminate between two hypotheses, the null ($\Hnull$) that the bandwidth is idle versus the alternative ($\Halt$) that the bandwidth is occupied, given a set of observations.
Based on the decision regarding the presence or absence of primary user's activities, the cognitive radio can utilise the spectrum or vacate it. Since we will formulate the problem of deciding whether the channel is occupied or not based on the observational evidence via a nested model structure, we can consider the likelihood ratio test framework, see \cite{vantrees1968dea}. 

In this paper, the challenging aspect of this problem, that extends beyond solutions developed previously, is that the probability distribution (PDF) of the test statistic that we wish to use for inference is not known in closed form. In general, it will also depend on a set of unknown parameters, for each hypothesis. Therefore we resort to formulating analytic approximate solutions for the distribution of the LRT under both hypotheses in order to perform inference.
\vspace{-0.5cm}
\subsection{Statistical model} \label{Statistical_model}
The network architecture we consider is a centralized network entity such as a base-station in infrastructure-based networks (see Fig. \ref{fig:system_model}). 
We consider a frame by frame scenario where one PBS may be active (transmitting data) or idle (not transmitting) during a frame. If active, its signal is transmitted over independent wireless channels and is captured by $M$ relay links. 
Each relay, instead of making individual decisions about the presence of the primary user, simply transmits the noisy received signal to the SBS  over a fading channel. The SBS is equipped with $N$ receive antennas. 
We further assume that the SBS has only limited knowledge of the CSI (noisy channel estimates), which is a practical scenario \cite{chen2011}. We now outline the system model and associated assumptions.\\
\textbf{Model Assumptions:}
\begin{enumerate}
\item \textsl{Assume a wireless network with one PBS equipped with a single antenna, that may be \textbf{active} or \textbf{idle} in a given frame.}
\item \textsl{In case that the PBS is \textbf{active}, it periodically transmits pilot signals, $s(l), \;l=1,\ldots,L$, within a frame of $L$ symbols, see \cite{IEEE802_22}. This model assumption will be discussed in Remark 2 below.}
\item \textsl{At each frame the received signal at the $m$-th relay ( $m=1,\ldots,M$) is a random variable given as a composite model, where $\Hnull$ and $\Halt$ correspond to \textbf{idle} and \textbf{active} model hypotheses, respectively:}
\begin{align}
\label{composite_model}
\begin {cases}
\begin{split}
    &\Hnull: R_m(l)=V_m(l)            \hspace{2.2cm}\; l=1,\ldots,L\\
    &\Halt: R_m(l)=F_m(l) s(l)+V_m(l)\;\hspace{0.2cm}l=1,\ldots,L,
\end{split}
\end {cases} 
\end{align}
\textsl{where $F_m(l)$ denotes the channel coefficient between the PBS and the $m$-th relay and $V_m(l)$ is the unknown noise realization associated with the $m$-th relay receiver.
Note, each of the relays is equipped with single receive and single transmit antenna.
}
\item \textsl{The relays re-transmit their received signal, $\left\{R_m(l)\right\}_{m=1}^M$, over $M$ fading channels. These channels can either occupy the same frequencies as the PBS-Relay channels or be dedicated reporting channels.}
\item \textsl{The received signal at the SBS from all $M$ relays at epoch $l$ can be written as
\begin{align}
\label{lin_system}
\left[
\begin{array}{c}
Y^{(1)}\left(l\right)\\
Y^{(2)}\left(l\right)\\
\vdots\\
Y^{(N)}\left(l\right)
\end{array}
\right]
=
\left[
\begin{array}{cccc}
G^{(1,1)}\left(l\right) & G^{(1,2)}\left(1\right)& \cdots& G^{(1,M)}\left(l\right)\\
G^{(2,1)}\left(l\right) & G^{(2,2)}\left(1\right)& \ddots& G^{(2,M)}\left(l\right)\\
\vdots&\ddots &\ddots&\vdots \\
G^{(N,1)}\left(l\right) & G^{(N,2)}\left(1\right)& \cdots& G^{(N,M)}\left(l\right)
\end{array}
\right]
\left[
\begin{array}{c}
R^{(1)}\left(l\right)\\
\vdots\\
R^{(M)}\left(l\right)
\end{array}
\right]+
\left[
\begin{array}{c}
W^{(1)}\left(l\right)\\
W^{(2)}\left(l\right)\\
\vdots\\
W^{(N)}\left(l\right)
\end{array}
\right]
\end{align}
and can be expressed compactly as the following composite model:
\begin{align}
\label{observation_model}
\begin {cases}
\begin{split}
    &\Hnull: \Y(l)=\G(l) \V(l) + \W(l)	\hspace{2.3cm}\; l=1,\ldots,L\\
    &\Halt: \Y(l)=\G(l) \left(\F(l)s(l)+\V(l)\right) + \W(l)\;l=1,\ldots,L,
\end{split}
\end {cases}
\end{align}
where $\Y(l) \in C^{N \times 1}$ is the received signal at the $l$-th sample, $\G(l) \in C^{N \times M}$ is the random channel matrix between the relay and the SBS, $\F(l)\triangleq \left[F_1(l),\cdots, F_M(l)\right]^T  \in C^{M \times 1}$ is the random channel vector between the PBS and the relays. The random vector, $\W(l)  \in C^{N \times 1}$, is the random additive noise at the SBS, and 
$\V(l)  \triangleq \left[V_1(l),\cdots, V_M(l)\right]^T \in C^{M \times 1}$ is the random additive noise at the relays.\\
}
\end{enumerate}
\vspace{-0.7cm}
\noindent \textbf{Remark 1:} \textit{Note, the dimension $N$ at the SBS can be attributed to several factors. For example, the SBS can be a MIMO receiver equipped with $N$ receive antennas. A different option is that the relays, while observing the same frequency band, transmit their information over $M$ dedicated orthogonal frequency bands (i.e. $N=M$). In that case, $\G$ would be a diagonal matrix. Here, we wish to make the system model as general as possible and not impose particular constraints or assumptions.}\\
\noindent \textbf{Remark 2:} \textit{Cognitive radio standard defined in IEEE $802.22$ is implemented in the TV bands \cite{IEEE802_22}. The TV bands digital signals can be either ATSC (North America), DVB-T (Europe), or ISDB (Japan). These standards contain within them many features, such as \textbf{pilot symbols and synchronization patterns}. For example, ATSC signals \cite{atsc74recommended} contain a $511$-symbol long PN sequence, pilot symbols and synchronization patterns of $828$ symbols. This makes our assumption regarding pilot symbols and synchronized transmission practical.}\\
\noindent	\textbf{Remark 3:} \textit{CSI can be obtained using the knowledge of the pilot symbols from Remark 2. If the relays have the capability of performing channel estimation, they can forward these estimates to the SBS. The SBS can also perform channel estimation to obtain matrix $\G$. 
}\\
%
\textbf{Prior specification:}\\
Here we present the relevant aspects of the Bayesian model and associated assumptions. 
\begin{enumerate}
\item \textsl{The PBS is \textbf{active} or \textbf{idle} with prior probabilities $P\left(\Halt\right)$ and $P\left(\Hnull\right)$, respectively.}

\item \textsl{All the channels are time varying, meaning that they are constant within a symbol, but may change from one symbol to the next.}
\item \textsl{The SBS has only a noisy estimate of the true channel realisation, $\G$. This is the result of a channel estimation phase which we do not detail here, see \cite{ding2009mimo}. A common approach is to model the channel as 
$\G(l)  = \overline{\G}(l)+ \Delta,$
where $\overline{\G}(l)$ is the noisy channel estimate, and $\Delta$ is the associated estimation error.
The distribution of $\G(l)$ conditioned on $\overline{\G}(l)$ and $\Delta$ can be written as
$\G(l) \sim CN\left(\ \overline{\G}(l), \Sigma_{\G} \right),$
$\Sigma_{\G} $ is the covariance matrix with known elements $\sigma^2_{\G}$, see details on noisy channel models in \cite{ding2009mimo}, \cite{rey2005robust}.}
\item \textsl{The SBS has only a noisy channel estimate of the true channel realisation, $\F(l)$. As with the $\G$ channels, $\F(l)$ can be written as
$\F(l) \sim CN\left(\ \overline{\F}(l), \Sigma_{\F} \right),$
where $\overline{\F}(l)$ is the estimated channel and $\Sigma_{\F}=\sigma^2_{\F} \I$ is the covariance matrix with known elements $\sigma^2_{\F}$.
Note: this stochastic model covers the case where CSI is unavailable, and only the channel prior distributions are available. In that case, $\overline{\F}(l)=\ZV$ and $\sigma^2_{\F}=1$.
}
\item \textsl{The additive noise at the relays is a zero-mean i.i.d Complex Gaussian distribution, 
$
\V(l) \sim CN \left( \mathbf{0}, \Sigma_{\V} \right),
$
where $\Sigma_{\V} = \sigma^2_{\V} \I $ is the covariance matrix, known at the SBS.
}
\item \textsl{The additive noise at the SBS is a zero-mean i.i.d Complex Gaussian distribution,
$
\W(l) \sim CN\left(\ \mathbf{0}, \Sigma_{\W} \right),
$
where $\Sigma_{\W} = \sigma^2_{\W} \I $ is the covariance matrix, known at the SBS.
}
\item \textsl{The symbols $s(l)$ are known at the SBS in the form of pilot symbols \cite{IEEE802_22}. For ease of presentation and w.l.o.g, we assume that $s(l)=1, \; \forall l \in \left\{1,\ldots,L\right\}$. }
\end{enumerate}

\subsection{Spectrum sensing decision criterion}
The objective of spectrum sensing is to make a decision whether the spectrum band is \textbf{idle} or \textbf{active} (choose $\Hnull$ or $\Halt$) in a given frame, based on the received signal at the SBS. To solve the decision problem, we will take a Bayesian approach. That is we consider the Bayes' risk formulation of the decision problem which generalises the LRT to a Bayesian framework. The problem of designing the decision rule can be treated as an optimization problem whose objective is to minimize the cost function:
\begin{align}
\begin{split}
C &=
P\left(\Hnull\right)
\left(C_{00}\int_{A_0}p\left(\y_{1:L}|\Hnull\right)d\y_{1:L}
+C_{10}\int_{A_1}p\left(\y_{1:L}|\Hnull\right)d\y_{1:L}\right)\\
&+P\left(\Halt\right)
\left(C_{01}\int_{A_0}p\left(\y_{1:L}|\Halt\right)d\y_{1:L}
+C_{11}\int_{A_1}p\left(\y_{1:L}|\Halt\right)d\y_{1:L}\right).
\end{split}
\end{align}
It can be shown, \cite{vantrees1968dea}, that the optimum decision rule is a likelihood-ratio test (LRT) given by 
\begin{equation}
\label{LRT}
\Lambda\left(\Y_{1:L}\right) \triangleq
\frac{p\left(\y_{1:L}|\Halt\right)}{p\left(\y_{1:L}|\Hnull\right) }
\begin{array}{c}
\stackrel{\Halt}{\geq} \\
\stackrel{<}{\Hnull}
\end{array}%
\frac{P\left(\Hnull\right)}{P\left(\Halt\right)}
\frac{C_{10}-C_{00}}{C_{01}-C_{11}}\triangleq \gamma,
\end{equation}
where $C_{xy}$ is the predefined associated cost of making a decision $\mathcal{H}_x$,  given that the true hypothesis is $\mathcal{H}_y$, and we define random matrix $\Y_{1:L} \triangleq \left[\Y\left(1\right),\ldots, \Y\left(L\right)\right]$.

Under both hypotheses $\Hnull$ and $\Halt$, all random quantities in the model are independent. We can therefore decompose the full marginals under each hypothesis, $p\left(\y_{1:L}|\mathcal{H}_k\right)$, $k={0,1}$, as
\begin{equation}
\label{marginal_likelihood}
p\left(\y_{1:L}|\mathcal{H}_k\right) 
= \prod_{l=1}^L p\left(\y(l)| \mathcal{H}_k\right).
\end{equation}
This decomposition is useful as it allows us to work on a lower dimensional space, resulting in efficiency gains for the algorithms we develop in the next sections and requiring no memory storage for data.

Table \ref{Table_algorithms} presents a summary of the different scenarios that will be covered in the next Sections as well as the type of solution that is provided under each scenario.
\section{Perfect Knowledge of PBS-relays and relays-SBS Channels} \label{PERFECT_CSI}
We consider the situation of perfect CSI of both $\G(l)$ and $\F(l)$ for all $l$, which corresponds to Case I in Table \ref{Table_algorithms}. 
We derive the optimal decision rule and the probabilities of detection and false alarm. This scenario enables us to obtain a lower bound on the overall system performance in terms of error probabilities in analytic form.
\begin{lemma}
\label{Lemma_1}
\textit{
The marginal likelihood under perfect CSI is:
\begin{align}
\label{marginal_likelihood_relay}
\Y(l) |\g(l),\f(l) \sim F(\y(l)|\g(l),\f(l)) \triangleq
\begin {cases} 
&CN \left( \ZV, \Sigma(l) \right), \hspace{0.5cm}\Hnull\\
&CN \left( \mu(l) , \Sigma(l)  \right),\Halt,
 	\end {cases}
\end{align}
where $\Sigma(l) \triangleq \sigma^2_{\V} \g(l)\g(l)^H + \sigma^2_{\W}\I $, and
$\mu(l) \triangleq \g(l) \f(l) $.
}
\end{lemma}
Next, using the decomposition property of (\ref{marginal_likelihood}), the test statistic and associated decision rule are presented.
\begin{theorem}
\textit{
\label{Theorem_3}
Under perfect CSI, the optimal decision rule defined in (\ref{LRT}) is given by
\begin{align}
\Lambda\left(\Y_{1:L}\right) \triangleq
\frac{p\left(\y_{1:L}|\Halt\right)}{p\left(\y_{1:L}|\Hnull\right)}
=\frac{\exp^{-\frac{1}{2}\sum_{l=1}^L\left(\y(l) - \mu(l)\right)^H \Sigma^{-1}(l) \left(\y(l) - \mu(l)\right) }}
{ \exp^{-\frac{1}{2}\sum_{l=1}^L\y(l)^H\Sigma^{-1}(l) \y(l) }},
\end{align}
which results in the following decision rule
\begin{align}
\begin{split}
  \log \gamma +
  \frac{1}{2} \sum_{l=1}^L \mu(l)^H \Sigma^{-1}(l)\mu(l)
\begin{array}{c}
\stackrel{\Hnull}{\geq} \\
\stackrel{<}{\Halt}
\end{array}%
 \sum_{l=1}^L
\text{Re}\left[\mu(l)^H \Sigma^{-1}(l)\Y(l)\right],
\end{split}			 
\end{align}
where we identify the test statistics according to
\begin{align}
\mathbb{T}(\Y_{1:L}) \triangleq 
\sum_{l=1}^L \text{Re}\left[\mu(l)^H \Sigma^{-1}(l)\Y(l)\right],
\end{align}
and the threshold for the critical region is given by
\begin{align}
\Gamma \triangleq  
\log \gamma +  \frac{1}{2} \sum_{l=1}^L \mu(l)^H \Sigma^{-1}(l)\mu(l).
\end{align}
}
\end{theorem}
\normalsize
\begin{proof}
Using Lemma \ref{Lemma_1} and the definition of the LRT it follows in the $\log$ domain that 
\begin{align}
\begin{split}
\log \Lambda\left(\Y_{1:L}\right) &= \frac{1}{2}\sum_{l=1}^L \Y(l)^H \Sigma^{-1}(l)\Y(l) - \frac{1}{2}\sum_{l=1}^L \left(\Y(l) -\mu(l)\right)^H \Sigma^{-1}(l) \left(\Y(l) -\mu(l)\right) \\
&=\sum_{l=1}^L
\text{Re}\left[\mu(l)^H \Sigma^{-1}(l)\Y(l)\right]
- \frac{1}{2} \mu(l)^H \Sigma^{-1}(l)\mu(l).
\end{split}			 
\end{align}
\end{proof}
This result is useful as it will provide a lower bound for the achievable Type I (false detection) and Type II (false alarm) probabilities as a function of SNR. It can therefore be used as a comparison for our approximation results when only partial CSI is known.
\begin{theorem}
\label{Theorem_1}
\textit{
The probability of detection and false alarm under perfect CSI are expressed analytically as
\begin{align}
\begin{split}
\label{p_d_csi}
p_d 
\triangleq 
p \left(\mathbb{T}(\Y_{1:L}) \geq  \Gamma| \Halt\right)
=
\mathcal{Q}\left[
\sqrt{2}
\frac{ \log \gamma -\frac{1}{2}\sum_{l=1}^L \mu(l)^H \Sigma^{-1}(l) \mu(l)}
{ \sqrt{\sum_{l=1}^L \mu(l)^H   \Sigma^{-1}(l) \mu(l)}}
\right],
\end{split}
\end{align}
and
\begin{align}
\begin{split}
\label{p_fa_csi}
p_f
\triangleq 
p \left(\mathbb{T}(\Y_{1:L}) \geq  \Gamma| \Hnull\right)
=
\mathcal{Q}\left[	
\sqrt{2}\frac{ \log \gamma +\frac{1}{2}\sum_{l=1}^L \mu(l)^H \Sigma^{-1}(l) \mu(l)}
{ \sqrt{\sum_{l=1}^L\mu(l)^H   \Sigma^{-1}(l) \mu(l)}}
\right],
\end{split}
\end{align}
respectively, where
$\mathcal{Q}\left[x	\right] = \frac{1}{\sqrt{2 \pi}}\int_x^{\infty} \exp^{-t^2/2} dt$.
}
\end{theorem}
\begin{proof}
To obtain this we derive the distribution of the test statistic utilised in the Bayes risk criterion under both the null and alternative hypotheses as follows
\small
\begin{align*}
\begin{split}
\mathbb{T}(\Y_{1:L})|\Hnull&\sim
N\left(\ZV,
\frac{1}{2}
\sum_{l=1}^L 
\text{Re}\left[\left(\mu(l)^H \Sigma(l)^{-1}\right)\Sigma(l) \left(\mu(l)^H \Sigma(l)^{-1}\right)^H\right]
 \right)
 =
N\left(\ZV,
\frac{1}{2}
\sum_{l=1}^L 
\mu(l)^H   \Sigma(l)^{-1} \mu(l)
 \right),  
\end{split}			 
\end{align*}
\begin{align*}
\begin{split}
\mathbb{T}(\Y_{1:L})|\Halt&\sim
N\left( 
\sum_{l=1}^L \text{Re}\left[\mu(l)^H \Sigma(l)^{-1} \mu(l)\right],
\frac{1}{2}
\sum_{l=1}^L 
\text{Re}\left[
\left(\mu(l)^H \Sigma(l)^{-1}\right)\Sigma(l) \left(\mu(l)^H \Sigma(l)^{-1}\right)^H
\right]
\right)\\
&=N\left( \sum_{l=1}^L \mu(l)^H \Sigma(l)^{-1} \mu(l),
\frac{1}{2}
\sum_{l=1}^L 
\mu(l)^H   \Sigma(l)^{-1} \mu(l)
\right).
\end{split}			 
\end{align*}
\normalsize
\end{proof}
Next we establish that adding receive antennas translates into a better overall detection performance.  
\begin{theorem}
\textit{
\label{theorem_MSE}
Under perfect CSI, the probability of false alarm, $p_f$, and misdetection, $1- p_d$, can be shown to be monotonically strictly decreasing as the number of SBS antennas $N$ increases.
}
\end{theorem}
\begin{proof}
To obtain this result, we apply Theorem \ref{Theorem_3} and show that $p_d(N+1) < p_d(N)$. Consequently, we prove that
 $ \mu_N^H \Sigma_N^{-1} \mu_N < \mu_{N+1}^H \Sigma_{N+1}^{-1} \mu_{N+1}$, where the subscript $N$, ($N+1$) refers to the number of receive antennas. We begin by proving that
\begin{align}
\g^H_N \left(\sigma^2_{\V} \g_N\g_N^H + \sigma^2_{\W}\I_N\right)^{-1} \g_N
\prec
\g^H_{N+1} \left(\sigma^2_{\V} \g_{N+1}\g_{N+1}^H + \sigma^2_{\W}\I_{N+1}\right)^{-1} \g_{N+1}.
\end{align}
Consider the following linear model
\begin{align}
\Z_N =  \g_N \X +\W ,
\end{align}
where $\W \sim CN\left(\ZV, \sigma^2_{\W}\I_N\right)$, $\X \sim CN\left(\ZV, \sigma^2_{\V}\I_K\right)$, and known mixing matrix $\g_N \in C^{N \times K}$. Then, using the properties of the Bayesian MMSE, its MSE covariance matrix, $\c_N$, can be written as
\begin{align}
\c_N = \sigma^2_{\V}\I_K- \sigma^4_{\V} \g^H_N \left(\sigma^2_{\V} \g_N\g_N^H + \sigma^2_{\W}\I_N\right)^{-1} \g_N.
\end{align}
If we had an augmented model with $(N+1)$ observations, i.e. $\Z_N \in C^{(N+1) \times 1}, \g_N \in C^{(N+1) \times M}, \W_N \in C^{(N+1) \times 1}$, then 
\begin{align}
\c_{N+1} = \sigma^2_{\V}\I_K- \sigma^4_{\V} \g^H_{N+1} \left(\sigma^2_{\V} \g_{N+1}\g_{N+1}^H + \sigma^2_{\W}\I_{N+1}\right)^{-1} \g_{N+1}.
\end{align}
It is well known that the MSE covariance is strictly decreasing with the number of observations \cite{vantrees1968dea}, that is
\begin{align}
\x^H \g^H_N \left(\sigma^2_{\V} \g_N\g_N^H + \sigma^2_{\W}\I_N\right)^{-1} \g_N \x
<
\x^H \g^H_{N+1} \left(\sigma^2_{\V} \g_{N+1}\g_{N+1}^H + \sigma^2_{\W}\I_{N+1}\right)^{-1} \g_{N+1} \x,
\end{align}
and in particular, consider $\x = \f \in C^{M \times 1}$.
\end{proof}

\textbf{Remark 5:} \textit{Theorem \ref{theorem_MSE} does not hold for the number of relays, i.e. it's not necessarily true that for fixed $N$ and $L$, $p_d(M+1) < p_d(M)$.}
\section{Imperfect PBS-relays CSI and Perfect relays-SBS CSI} \label{Laguerre_section}
In this section we consider the system model under which we can assume that the receiver has perfect CSI of $\G(l)$ but only partial knowledge of $\F(l)$, which corresponds to Case II in Table \ref{Table_algorithms}. We derive the optimal decision rule and the probabilities of detection and false alarm via Laguerre series expansion.
\begin{lemma}
\textit{
\label{Lemma_4}
The marginal likelihood under Imperfect PBS-relays CSI and Perfect relays-SBS CSI is:
\begin{align}
\label{Imperfect_PBS_relays}
\Y(l) |\g(l) \sim F(\y(l) |\g(l)) \triangleq
\begin {cases} 
&CN \left( \ZV, \Sigma_{\mathcal{H}_0}(l) \right),\hspace{0.5cm} \Hnull\\
&CN \left( \mu(l) , \Sigma_{\mathcal{H}_1}(l)\right) ,\Halt,
 	\end {cases}
\end{align}
where 
$\Sigma_{\mathcal{H}_0}(l) \triangleq \sigma^2_{\V} \g(l)\g(l)^H + \sigma^2_{\W}\I \;$, $\; \Sigma_{\mathcal{H}_1}(l) \triangleq \left(\sigma^2_{\V}+\sigma^2_{\F}\right) \g(l)\g(l)^H + \sigma^2_{\W}\I$\\ 
and $\mu(l) \triangleq \g(l) \overline{\F}(l) $.
}
\end{lemma}
Having stated the likelihood in Lemma \ref{Lemma_4} we then present the corresponding test statistic for the Imperfect PBS-relays CSI combined with perfect relays-SBS CSI setting.
\begin{theorem}
\textit{
\label{Theorem_9}
Under Imperfect PBS-relays CSI and Perfect relays-SBS CSI, the optimal decision rule in (\ref{LRT}) is given by
\begin{align*}
\Lambda\left(\Y_{1:L}\right) \triangleq
\frac{p\left(\y_{1:L}|\Halt\right)}{p\left(\y_{1:L}|\Hnull\right)}
=\frac{\prod_{l=1}^L p\left(\y(l)|\Halt\right)}
			 {\prod_{l=1}^L p\left(\y(l)|\Hnull\right)}
=
\frac{
\prod_{l=1}^L 
\frac{1}{ \left(2\pi\right)^{N/2} \left|\Sigma(l)_{\mathcal{H}_1}\right|^{1/2}} 
\exp^{-\frac{1}{2}
\left(\y(l) - \mu(l)\right)^H 
\Sigma_{\mathcal{H}_1}^{-1} (l) 
\left(\y(l) - \mu(l)\right) }
}{ 
\prod_{l=1}^L
\frac{1}{ \left(2\pi\right)^{N/2}\left|\Sigma_{\mathcal{H}_0}(l)\right|^{1/2}}
\exp^{-\frac{1}{2}\y(l)^H\Sigma_{\mathcal{H}_0}^{-1} (l) \y(l) }},
\end{align*}
which results in the following decision rule
\begin{align*}
\begin{split}
& \log \gamma +
\sum_{l=1}^L 
\log \left|\frac{ \Sigma_{\mathcal{H}_1}(l)}
     {\Sigma_{\mathcal{H}_0}(l)} \right|
+
\sum_{l=1}^L
\left(
\a(l)^H \a(l)
+\mu(l)^H \Sigma_{\mathcal{H}_1}^{-1} (l)\mu(l)
\right)
\begin{array}{c}
\stackrel{\Hnull}{\geq} \\
\stackrel{<}{\Halt}
\end{array}%
\sum_{l=1}^L
\left|\c(l) \Y(l)+  \a(l)\right|^2.
\end{split}
\end{align*}
Here we identify the test statistics according to
\begin{align}
\label{test_statistics_perfect_imperfect}
\mathbb{T}(\Y_{1:L})= 
\sum_{l=1}^L
\left|\c(l) \Y(l)+ \a(l)\right|^2.
\end{align}
and the resulting Bayes risk threshold is defined as
\begin{align*}
\Gamma \triangleq  
\log \gamma +
\sum_{l=1}^L 
\log \left|\frac{ \Sigma_{\mathcal{H}_1}(l)}
     {\Sigma_{\mathcal{H}_0}(l)} \right|
+
\sum_{l=1}^L
\left(\a(l)^H \a(l)+\mu(l)^H \Sigma_{\mathcal{H}_1}^{-1} (l)\mu(l)
\right),
\end{align*}
with 
$\c(l)^H \c(l) \triangleq \left(\Sigma_{\mathcal{H}_0}^{-1} (l)-\Sigma_{\mathcal{H}_1}^{-1} (l)\right)$,
and $\a(l) \triangleq \c(l)^{-1} \Sigma_{\mathcal{H}_1}^{-1} (l) \mu(l)$.
}
\end{theorem}

\begin{proof}
Using the result in Lemma \ref{Lemma_4} and the definition of the LRT, produces
\small
\begin{align*}
\begin{split}
2 \log \Lambda\left(\Y_{1:L}\right) &
=
\sum_{l=1}^L 
\log \left|\frac{ \Sigma_{\mathcal{H}_0}(l)}
     {\Sigma_{\mathcal{H}_1}(l)} \right|
+
\sum_{l=1}^L 
\Y(l)^H \Sigma_{\mathcal{H}_0}^{-1} (l)\Y(l)
- \sum_{l=1}^L
\left(\Y(l) -\mu(l)\right)^H
\Sigma_{\mathcal{H}_1}^{-1} (l)
\left(\Y(l) -\mu(l)\right) \\
&=
\sum_{l=1}^L 
\log \left|\frac{ \Sigma_{\mathcal{H}_0}(l)}
     {\Sigma_{\mathcal{H}_1}(l)} \right|
+\sum_{l=1}^L
\left|\c(l) \Y(l)+\  \a(l)\right|^2
-\a(l)^H \a(l) -\mu(l)^H \Sigma_{\mathcal{H}_1}^{-1} (l)\mu(l).
\end{split}			 
\end{align*}
\end{proof}
The distribution of the test statistic in (\ref{test_statistics_perfect_imperfect}) is asymptotically $\chi^2$ in $L$.
However, in practical systems the number of frames is typically small and therefore this asymptotic result can not be applied. Therefore, in these cases the distribution of the test statistic in (\ref{test_statistics_perfect_imperfect}) is not attainable in closed form and deriving the probability of detection and false alarm needs to be approximated. The reason for this is that the $L$-fold convolution of non-central $\chi^2$ random variables, each with different centrality parameter, can not be solved in closed form.

There does however exist a rich statistical literature approximating the distribution of the linear combination of non-central $\chi^2$ random variables. The solutions to finding an approximation to the PDF and CDF of a linear combination of non-central $\chi^2$ involve a range of series expansions, saddle point approximation type methods and the Weiner Germ modifications, see in depth discussions in \cite{penev2000weiner} and \cite{kotz2000continuous}. In this paper we consider the Laguerre series expansion distributional approximations to attain the probability of false alarm and mis-detection. This class of approximation has tight error bounds represented as a function of the order of the series, see \cite{castano2005distribution}.

The Laguerre series expansion for the probability of false alarm and mis-detection is characterized by the parameters: $p$ the order of the series expansion; $\mu_0$ a parameter that controls the rate of convergence of the series expansion; and $\beta$ values selected to control the error of approximation for a given $p$, see discussion in \cite{castano2005distribution}. In addition, this approximation has the property that for different settings of the parameter $\mu_0$ we can obtain other series expansions in the literature such as setting $\mu_0 = \nu/2 = p$ which gives the expansion of \cite{kotz2000continuous}. 
\begin{theorem}
\textit{
\label{theorem_Laguerre}
Under Imperfect PBS-relays CSI and Perfect relays-SBS CSI, the probability of false alarm and miss-detection are approximated by the generalized Laguerre series as:
\begin{subequations}
\begin{align}
\label{p_d_csi_laggure}
\widehat{p}_d &= p\left(\mathbb{T}(\Y_{1:L})\geq \Gamma |\Halt\right) = 1- \widehat{F}^{p}(\Gamma  |\Halt) \\
\label{p_f_csi_laggure}
\widehat{p}_f &= p\left(\mathbb{T}(\Y_{1:L}) \geq \Gamma |\Hnull\right) = 1- \widehat{F}^{p}(\Gamma  |\Hnull),
\end{align}
\end{subequations}
where $\Gamma$ is the Bayes risk threshold given in Theorem \ref{Theorem_9} and
\small
\begin{align}
\mathbb{T}(\Y_{1:L}) |\mathcal{H}_k \sim 
\widehat{F}^{p}(t=\mathbb{T}(\Y_{1:L})|\mathcal{H}_k) = \frac{e^{-\frac{t}{2\beta}}}{\left(2\beta\right)^{\mu/2 + 1}} \frac{t^{\nu/2}}{\Gamma(\nu/2 + 1)} \sum_{k \geq 0}\frac{k!m_k}{\left(\nu/2 + 1\right)_k}L_k^{(\nu/2)}\left(\frac{(\nu+2) t}{4 \beta \mu_0}\right), \; \; \forall \mu_0>0, t \in \mathcal{R},
\end{align}
\normalsize
with coefficients $m_k$ having the recurrence relations in the setting $\mu_0>0$ and $p=\nu/2 + 1$ given by
\small
\begin{align*}
m_0 &= 2\left(\frac{\nu}{2} + 1\right)^{\nu/2 + 1} \exp^{\left(-\frac{1}{2}\sum_{l=1}^L \frac{\delta(l) \alpha_l (p-\mu_0)}{\beta \mu_0 + \alpha_l(p-\mu_0)}\right)} \frac{\beta^{\mu/2 + 1}}{p-\mu_0}\prod_{l=1}^L\left(\beta \mu_0 + \alpha_l(p-\mu_0)\right)^{-\nu(l)/2},\\
m_k &= \frac{1}{k}\sum_{j=0}^{k-1}m_j d_{k-j} \; \; \; k \geq 1,\\
d_j &= -j \frac{\beta p}{2 \mu_0} \sum_{l=1}^L \delta(l) \alpha_l (\beta - \alpha_l)^{j-1}\left( \frac{\mu_0}{\beta \mu_0 + \alpha_l (p-\mu_0)}\right)^{j+1} + \left(\frac{-\mu_0}{p-\mu_0}\right)^j + \sum_{l=1}^L\frac{\nu(l)}{2}\left(\frac{\mu_0(\beta-\alpha_l)}{\beta \mu_0 + \alpha_l(p-\mu_0}\right)^j, \; \; j \geq 1.
\end{align*}
\normalsize
The corresponding PDF is given by
\begin{align}
\widehat{f}^{p}(t=\mathbb{T}(\Y_{1:L})|\mathcal{H}_k) = \frac{e^{-\frac{t}{2\beta}}}{\left(2\beta\right)^{\mu/2}} \frac{t^{\nu/2-1}}{\Gamma(\nu/2)} \sum_{k \geq 0}\frac{k!c_k}{\left(\nu/2\right)_k}L_k^{(\nu/2-1)}\left(\frac{\nu t}{4 \beta \mu_0}\right), \; \; \forall \mu_0>0, t \in \mathcal{R},
\end{align}
with $p=\nu/2$, and $\nu = \sum_{l=1}^L \nu(l)$ and the following recurrence relations for the coefficients,
\begin{align*}
\begin{split}
c_0 &= \left(\frac{\nu}{2\mu_0}\right)^{\nu/2}\exp^{\left(-\frac{1}{2}\sum_{l=1}^L \frac{\delta(l) \alpha_l (p-\mu_0}{\beta \mu_0 + \alpha_l(p-\mu_0}\right)}\prod_{l=1}^L\left(1+\frac{\alpha_l}{\beta}(p/\mu_0 - 1)\right)^{-\mu(l)/2},\\
c_k &= \frac{1}{k}\sum_{j=0}^{k-1}c_j d_{k-j},\\
d_j &= -j \frac{\beta p}{2 \mu_0} \sum_{l=1}^L \delta(l) \alpha_l (\beta - \alpha_l)^{j-1}\left( \frac{\mu_0}{\beta \mu_0 + \alpha_l (p-\mu_0)}\right)^{j+1}  +\sum_{l=1}^L\frac{\nu(l)}{2}\left(\frac{1-\alpha_l\beta}{1+(\alpha_l\beta)(p/\mu_0 - 1)}\right)^j, \; \; j \geq 1,
\end{split}
\end{align*}
\normalsize
and $L_j^{(\alpha)}(t) = \sum_{m=0}^{j}C^{j+\alpha}_{j-m}\frac{(-t)^{m}}{m!} \; \;, \alpha>0$ is the generalized Laguerre polynomial. In addition the generalized Laguerre polynomials can be obtained by recurrence relationships,
\begin{align*}
\begin{split}
jL_j^{(\alpha)}(t) &= \left(2j + \alpha - 1 - t\right)L_{j-1}^{(\alpha)}(t) - \left(j+\alpha-1\right)L_{j-2}^{(\alpha)}(t),\\
L_{-1}^{(\alpha)}(t) &= 0,\;\;
 L_{0}^{(\alpha)}(t) = 1.
\end{split}
\end{align*}
\normalsize
}
\end{theorem}
\begin{proof} To derive this result involves consideration of the distribution for a linear combination of non-central $\chi^2$ random variables. Consider the identity for LRT statistic given in Theorem \ref{theorem_Laguerre} as
\begin{align}
\mathbb{T}(\Y_{1:L})= 
\sum_{l=1}^L ||\c(l) \Y(l)+ \a(l)||^2 = \sum_{l=1}^L ||\widetilde{\Y}(l)||^2,
\label{teststatLaguerre}
\end{align}
with 
\begin{subequations}
\begin{align}
\widetilde{\Y}(l) |\Hnull &\sim 
CN\left(\a(l), \underbrace{\c(l) \Sigma_{\mathcal{H}_0} \c(l)^H}_{\widetilde{\Sigma}_{\mathcal{H}_0}}\right),\\ 
\widetilde{\Y}(l) |\Halt &\sim
 CN\left(\c(l)\mu(l) + \a(l),\underbrace{\c(l) \Sigma_{\mathcal{H}_1}\c(l)^H}_{\widetilde{\Sigma}_{\mathcal{H}_1}}\right).
\end{align}
\end{subequations}
To obtain the distributional approximations, we require a linear combination of independent $\chi^2$ non-central random variables. To achieve this for each symbol we apply the following rotational transformation, based on SVD decomposition of $\widetilde{\Sigma}_{\mathcal{H}_k}(l)= U(l) \Lambda(l) U^T(l)$
giving transformed  random vectors with i.i.d elements $\Z(l)=\U(l)\Lambda^{-1/2}(l)\widetilde{\Y}(l) \sim CN(\U(l)\Lambda^{-1/2}(l) \a(l), \I)$. As a result, we obtain a univariate linear combination of squared Gaussian random variables,
\begin{align}
\mathbb{T}(\Z_{1:L})=\sum_{l=1}^{LN} \alpha_l Z^2(l),
\end{align}
with $\alpha_l>0$ a positive weight. Each resulting independent scalar random variable $Z^2(l) \sim \chi_{\nu_l}^2(\delta(l))$ with non-centrality parameters $\delta(l)$. 

Therefore, under the transformed observation vectors $\Z_{1:L}$, one can obtain the distributions of the test statistic in Theorem \ref{theorem_Laguerre}.
The equivalent Bayes risk threshold for the transformed data can be easily obtained by replacing $\c(l)$ with $\widetilde{\c}(l) = \U(l)\Lambda^{-1/2}(l) \c(l)$ and replacing $\a(l)$ with
$\widetilde{\a}(l) = \U(l)\Lambda^{-1/2}(l) \a(l)$.
\end{proof}

\vspace{0.5cm}
The result of Theorem \ref{theorem_Laguerre} provides the means to approximate the critical region of the decision rule for any observed test statistics for any number of frames. This means we can quantify the mis-detection and false alarms rates analytically as a function of the number of frames and the SNR with known error bounds on the order of approximation.
\section{Imperfect PBS-relays CSI and Imperfect relays-SBS CSI}\label{App_Algorithms}
In this section we consider the case where the SBS has only partial CSI of both $\G(l)$ and $\F(l)$, which corresponds to Case III and Case IV in Table \ref{Table_algorithms}. 

We first consider two scenarios which have practical interpretations before moving to the more general case. The first involves consideration of line-of-sight transmissions; and the second assumes high SNR scenario, which for both we obtain analytic expression for the marginal likelihood and therefore an analytic for the LRT in (\ref{LRT}).

For the $n$-th element in (\ref{lin_system}), after omitting the time dependency $l$, we obtain:
\begin{equation}
\label{k_term}
Y^{(n)}= \sum_{m=1}^M G^{(n,m)}R^{(m)} + W^{(n)} ,
\end{equation}
The $m$-th term in the summation above can be expressed as:
\small
\begin{equation}
\begin{split}
\label{product_normal}
G^{(n,m)}R^{(m)} &= \left( \Real  \left[ G^{(n,m)}\right]+j \Imag\left[G^{(n,m)}\right]\right)
 									 \left(\Real\left[R^{(m)}\right]+j \Imag\left[R^{(m)}\right]\right)\\
&=\Real\left[G^{(n,m)}\right]\Real\left[R^{(m)}\right]+
	j \Real\left[G^{(n,m)}\right]\Imag\left[R^{(m)}\right]+
	j \Imag\left[G^{(n,m)}\right]\Real\left[R^{(m)}\right]-
	\Imag\left[G^{(n,m)}\right]\Imag\left[R^{(m)}\right].
\end{split} 									 
\end{equation}
\normalsize
Each of the terms in (\ref{product_normal}) form a product of independent Normal random variables.
The distribution of this product was first derived by \cite{wishart2008distribution} and later studied by \cite{craig1936} and the resulting density and Moment Generating Function (MGF) are given as follows.
\begin{lemma}
\textit{
The distribution of a product of two independent normally distributed variates $Z=X Y$, where
$X\sim N\left(\overline{X}, \sigma^2_X\right)$ and $Y\sim N\left(\overline{Y}, \sigma^2_Y\right)$ is the solution of the following integral
\begin{equation}
\begin{split}
\label{product_normal_distribution}
p \left(z\right) =
\frac{1}
{2 \pi \sigma_X \sigma_Y}
 \int_{- \infty}^{ \infty}
 \int_{- \infty}^{ \infty} 
 \exp^{-\frac{x^2}{2 \sigma^2_X}}
 \exp^{-\frac{y^2}{2\sigma^2_Y}}
 \delta\left(z-x y\right)
 \text{d}x
 \text{d}y.
\end{split}
\end{equation}
The Moment Generating Function (MGF) of $Z$ can be expressed as \cite{craig1936}
\begin{equation}
\label{MGF}
M_{Z} \left(t\right)  = 
\frac{\exp\left\{
\frac{\left(\rho_X^2+\rho_Y^2\right)t^2+2 \rho_X \rho_Y t}
{2\left(1-t^2\right)}\right\}}
{\sqrt{1-t^2}},
\end{equation}
where $\rho_X = \frac{\overline{X}}{\sigma_X }$ and $\rho_Y = \frac{\overline{Y}}{\sigma_Y }$.\\
For the special case where $\overline{X}=\overline{Y}=0$ the integral can be solved analytically as \cite{craig1936}
\begin{equation}
\label{craig_distribution}
\begin{split}
p \left(z\right) =
 \frac {
 \text{K}_0
 \left( \frac{\left|z\right|}
 { \sigma^2_X \sigma^2_Y }
 \right)
 }
 {
 \pi \sigma_X \sigma_Y },
\end{split}
\end{equation}
where $\text{K}_0\left(\cdot\right)$ is the modified Bessel function of the second kind. 
}
\end{lemma}

Following this result of the MGF we obtain an asymptotic result for the marginal likelihood in (\ref{marginal_likelihood}) under both hypotheses.
\begin{theorem}
\textit{
\label{theorem_normal_product}
The distribution of $Y^{(n)}$ in (\ref{k_term}) is asymptotically Normal when either $\rho_G = \frac{\overline{G}^{(n,m)}}{\sigma_G } \rightarrow \infty   $ or $\rho_F = \frac{\overline{F}^{(m)}}{\sigma_F } \rightarrow \infty$. Therefore
\begin{equation}
\label{normal_PSCI}
Y^{(n)}= \sum_{m=1}^M G^{(n,m)}R^{(m)} + W^{(k)}\sim
\begin {cases} 
&CN \left( \mu_{\Hnull} , \Sigma_{\Hnull} \right), \Hnull\\
&CN \left( \mu_{\Halt} , \Sigma_{\Halt}\right) ,\Halt.
 	\end {cases}
\end{equation}
Therefore, assuming only the linear dependence between the $N$ components of $\Y$ and ignoring any tail dependence in the joint multivariate distribution, we conclude that $\Y$ is multivariate Gaussian.
}
\end{theorem}
\begin{proof}
See Appendix.
\end{proof}
Theorem \ref{theorem_normal_product} shows that under the following conditions, the Gaussian Approximation (GA), that will be presented next, is valid:
\begin{enumerate}
	\item The CSI estimation error, quantified by $\sigma^2_{\G}$ and/or $\sigma^2_{\F}$ is low.
  \item The mean value of one or both of the channels estimates ($\overline{G}^{(k,m)}$, $\overline{R}^{(m)}$) is large, i.e. strong line-of-sight, for example in Rician channels.
\end{enumerate}

Next we consider generalising the analysis to relax the assumptions in Theorem \ref{theorem_normal_product} making the resulting solution widely applicable. Consequently, the distribution of the marginal likelihood in (\ref{marginal_likelihood}) under both hypotheses is intractable. This is because it involves finding the distribution of $Y^{(n)}$ in (\ref{k_term}) which can not be obtained analytically. This is due to the fact that $\left(M_Z\left(t\right)\right)^M \neq M_{\sum_{m=1}^M Z_m}\left(t\right)$, which means that this distribution is not closed under convolution.
\subsection{Gaussian Approximation via Moment Matching}\label{Gaussian_Approximation}
We derive a low-complexity detection algorithm that is based on moment matching so that the distribution of the received signal is approximated by a matrix variate Gaussian distribution based on the results obtained in Theorem \ref{theorem_normal_product}. We show under which conditions this approximation is valid and asses the approximation error.

\begin{lemma}
\textit{
\label{Gaussian_apprx}
The first two moments of $\Y(l)$ can be expressed as
\begin{align*}
\begin{split}
\exE\left[\Y(l)\right] = \exE\left[\G(l) \R(l) + \W(l)\right] = 
\begin {cases} 
&\ZV, \hspace{1.3cm} \Hnull\\
&\overline{\G}(l)\; \overline{\R}(l), \Halt.
\end {cases}
\end {split} 
\end{align*}
\begin{align*}
\begin{split}
\exE \left[\Y(l) \Y(l)^H\right] &=
\exE\left[\left(\G(l) \R(l)+\W(l)\right)\left(\G(l) \R(l)+\W(l)\right)^H\right]  \\
&=\begin {cases} 
& M \sigma^2_{\V} \sigma^2_{\G}\I + \sigma^2_{\W}\I,  \hspace{4cm} \Hnull\\
&\sigma^2_{\G} \text{Tr} \left[ \b^H(l) \right]\ \I +\overline{\G}(l)\; \b(l) \overline{\G}^H(l) + \sigma^2_{\W}\I , \Halt
\end {cases}
\end {split} 
\end{align*}
where $\text{Tr}\left[X\right]$ is the trace of matrix $X$ and $\b(l) \triangleq \left(\Sigma_{\V}+\Sigma_{\F}+\overline{\F}(l)\;\ \overline{\F}^H(l)\right)$.\\
We make a Gaussian approximation on the multivariate observation vector to obtain:
\begin{align*}
\begin{split}
\Y(l) \sim
\begin {cases} 
&CN\left(\Y(l);\ZV, \underbrace{M \sigma^2_{\V} \sigma^2_{\G}\I + \sigma^2_{\W}\I}_{\Sigma_{\mathcal{H}_0(l) }}\right),  \hspace{0.3cm} \Hnull\\
&CN\left(\Y(l);\underbrace{\overline{\G}(l) \; \overline{\R}(l)}_{\mu(l)}, 
\underbrace{
\sigma^2_{\G} \text{Tr} \left[ \b^H(l) \right]\ \I +\overline{\G}(l)\; \b(l) \overline{\G}^H(l) 
-\mu(l) \mu^H(l) + \sigma^2_{\W}\I}_{\Sigma_{\mathcal{H}_1(l)}}\right), \Halt.
\end {cases}
\end {split} 
\end{align*}
\normalsize
}
\end{lemma}
The approximated distribution of $\Y$ matrix has the same structure as (\ref{Imperfect_PBS_relays}), and we can therefore utilise a similar procedure to obtain the decision rule:
\begin{lemma}
\textit{
\label{Lemma_9}
Utilizing the results in Lemma \ref{Gaussian_apprx} combined with Lemma \ref{Lemma_4} results in the LRT decision rule and Bayesian threshold as in Theorem \ref{Theorem_9}.
}
\end{lemma}
In making the GA, it is important to quantify the associated error with such a distributional assumption in evaluation of the LRT. Understanding the approximation error allows us to provide guidance on system design relating to the number of relay and the length of frames in order to mitigate errors in evaluating mis-detection and false alarms probabilities.

%
\begin{theorem}
\textit{
\label{theorem_Berry_Esseen}
Under a Gaussian approximation to the distribution of the linearly transformed received signals $\widetilde{\Y}(1),\ldots,\widetilde{\Y}(M)$, where 
$\widetilde{\Y}(m) = \mathbb{T}\left(\G^{(:,m)}R^{(m)}\right),\; m=1,\cdots,M$, and $\mathbb{T}\left(\cdot\right)$ is the linear standardization transformation, we obtain the Kolmogorov distance on all convex sets $A \in \mathcal{A}$ for\\
$\bm{S}_M = \frac{\widetilde{\Y}(1)+\ldots+\widetilde{\Y}(M)}{\sqrt{M}}$ given by
\begin{align*}
sup_{A \in \mathcal{A}}|Pr\left(\bm{S}_M \in A\right) - Pr\left(\Z \in A\right)| \leq 
\frac{400 N^{1/4} \mathbb{E}\left[||\widetilde{\Y}(m)||^3\right]  }
{\sqrt{M} },
\end{align*}
where $\mathbb{E}\left[||\widetilde{\Y}(m)||^3\right] = 2 \sqrt{2} \frac{\Gamma\left(\frac{N+3}{2}\right)}{\Gamma\left(\frac{N}{2}\right)} $, and $\Z \sim N\left(\bm{0},\I\right)$,
}
\end{theorem}
\begin{proof}
Using the result of Lemma \ref{Gaussian_apprx}, we transform the corresponding observation vectors $\Y(m)$ according to the following SVD decomposition of the covariance matrix of $\Y(m)$ given by $\Sigma(m) = \U(m)\Lambda(m) \U(m)^H$. This produces the transformed i.i.d. random vectors given by\\
$\widetilde{\Y}(m)=\U(m)\Lambda^{-1/2}(m) \left(\Y(m) - \mathbb{E}[\Y(m)]\right)$.
Having obtained i.i.d. vectors, we apply the multi-dimensional Berry-Essen bound \cite{dasgupta}. 
To do this we only need to calculate $\mathbb{E}\left[||\widetilde{\Y}(m)||^3\right]$.
Writing $||\widetilde{\Y}(m)|| = \sqrt{\left(\widetilde{\Y}^{(1)}(m)\right)^2+\ldots+\left(\widetilde{\Y}^{(N)}(m)\right)^2}$, we have that $\left(\widetilde{\Y}^{(i)}(m)\right)^2 \sim \chi_1^2 \left(0\right)\; \forall i\in\left\{1,\ldots,N\right\}$, and therefore
$\sum_{n=1}^N \left(\widetilde{\Y}^{(n)}(m)\right)^2 \sim \chi_N^2 \left(0\right)$ and consequently, 
$||\widetilde{\Y}(m)|| \sim \chi_N \left(0\right)$. The third moment of $||\widetilde{\Y}(m)||$, which follows a $\chi_N$ distribution, is given by
$2 \sqrt{2} \frac{\Gamma\left(\frac{N+3}{2}\right)}{\Gamma\left(\frac{N}{2}\right)}$. 
\end{proof}
\textbf{Remark}: considering convex sets of the form $\left(-\infty, \x \right]$, $\forall \x \in \mathbb{R}^N$, the Berry-Esseen result shows the maximum error we can make under our Gaussian approximation of each observation vector and therefore provides a bound on the approximation error on the marginal likelihood used in the LRT.\\
\textbf{Remark}: the maximum error we can make under GA decreases at a rate of $\sqrt{M}$ (that is, the number of relays) for a fixed number of antennas, $N$. Furthermore, for a fixed number of relays, $M$, the approximation error becomes unbounded for increasing number of receive antennas, since
$\frac{\Gamma\left(\frac{N+3}{2}\right)}{\Gamma\left(\frac{N}{2}\right)} \stackrel{N \rightarrow \infty}{\longrightarrow}\infty$.

\subsection{Approximation of the Marginal Evidence via Laplace Approach}
In this section, a more accurate estimation of the marginal likelihoods than the GA is developed.
This is based on the Laplace approximation \cite{kass1995bayes}. The Laplace method can approximate integrals via a series expansion which uses local information about the integrand around its maximum. Therefore, it is most useful when the integrand is highly concentrated in this region. 

Under the full Bayesian paradigm, the evidence in (\ref{marginal_likelihood}) is obtained via the following marginalisation:
\begin{equation}
\begin{split}
\label{marginal_likelihood_1_sample}
p\left(\y |\mathcal{H}_k\right) 
&= 
\int  p\left(\y |\r , \mathcal{H}_k\right)
 p\left(\r | \mathcal{H}_k\right)d  \r \\
 &= 
 \int_{r_1} \ldots \int_{r_M}
  p\left(\y |r_{r_1} , \mathcal{H}_k\right)
  \cdots
   p\left(\y |r_{r_M} , \mathcal{H}_k\right)
 p\left(r_1 | \mathcal{H}_k\right)
  \cdots
 p\left(r_M | \mathcal{H}_k\right)
 \text{d}  r_1 \ldots \text{d}  r_M ,
\end{split}
\end{equation}
where $\R  \triangleq \left[R_1\left(l\right),\ldots, R_M\left(l\right)\right]^T$.
The densities in (\ref{marginal_likelihood_1_sample}) can be expressed as:
\begin{align}
\label{distributions_1}
\Y |\left(\R=\r  ; \mathcal{H}_k \right) \sim CN \left(\ \overline{\G}  \r , \left(\sigma^2_{\G}\left\|\r \right\|^2+\sigma^2_{\W}\right) \I\right)
\end{align}
\begin{align}
\label{distributions_2}
\begin {split} 
\R \sim F(\r) \triangleq
\begin {cases} 
&CN \left(\ \ZV, \Sigma_V\ \right), \hspace{1.3cm} \Hnull\\
&CN \left(\ \overline{\F} s , \Sigma_{\F}+\Sigma_V\ \right), \Halt.
\end {cases}
\end {split} 
\end{align}
This integral is intractable and we shall approximate it via an application of the Laplace approximation.

To do so we begin by defining the following quantity (we discard the time dependency $l$ here):
\begin{align}
h\left(\r \right) \triangleq \log \left( p\left(\y |\r\right) p\left(\r\right)\right). 
\end{align}  
This expression is now expanded using a Taylor series about its maximum \textit{a-posteriori} (MAP) estimate, denoted by $\widehat{\R} = \argmax_{\r} p(\r|\y)$. This is the point where the posterior density is maximised, i.e. the mode of the posterior distribution. Hence, we obtain
\small
\begin{align}
\label{Taylor_expansion}
h\left(\r \right) = 
h\left(\widehat{\R} \right) + 
\left(\r  - \widehat{\R} \right)^T \underbrace{\frac{\partial h\left(\widehat{\R} \right)}{\partial \r }}_{
\left(=0\right) \text{ at MAP location}}+
\frac{1}{2} \left(\r  - \widehat{\R} \right)^T \frac{\partial^2 h\left(\widehat{\R} \right)}{\partial^2 \r }
\left(\r  - \widehat{\R} \right)+ \ldots \;\;.
\end{align} 
\normalsize
The second term in equation (\ref{Taylor_expansion}) cancels because at the maximum of $h\left(\r \right)$ (which is by definition what the MAP location represents), the first derivative is zero.\\
Replacing $h\left(\r \right)$  by the truncated second-order Taylor series yields:
\begin{align}
h\left(\r \right)\approx
h\left(\widehat{\R} \right) + 
\frac{1}{2} \left(\r  - \widehat{\R} \right)^H \mathbb{\H} \left(\r  - \widehat{\R} \right),
\end{align} 
where $\mathbb{\H}_k$ is the Hessian of the log posterior, evaluated at $\widehat{\R} $:
\begin{align}  
\label{Hessian}
\mathbb{\H}\triangleq  \left.\frac{\partial^2 h\left(\widehat{\R} \right)}{\partial^2 \r }\right|_{\r=\widehat{\R} }=
\left. \frac{\partial^2 \ln p\left(\r |\y \right)}{ \partial \r   \partial \r^H} \right|_{\r=\widehat{\R} }.
\end{align}  

We now concentrate on approximating the $\log$ of the integral in (\ref{marginal_likelihood_1_sample}):
\begin{align}
\label{LaplaceApp}
\begin{split}
\log p\left(\y \right)  
&= \log \int p\left(\y |\r\right) 
						p\left(\r\right) d \r\\
&= \log \int 						
						\exp^{h\left(\r \right)} d \r\\
&\stackrel{ \approx}{\tiny \text{Taylor series}} \log \int 						
						\exp^{h\left(\widehat{\r} \right) + 
\frac{1}{2} \left(\r  - \widehat{\R} \right)^T \mathbb{\H} \left(\r  - \widehat{\R} \right)} d \r\\
&= h\left(\widehat{\R} \right) +\log \int 						
						\underbrace{\exp^{\frac{1}{2} \left(\r  - \widehat{\R} \right)^T \mathbb{\H} \left(\r  - \widehat{\R} \right)}}
						_{\propto CN \left(\widehat{\R},\mathbb{\H}\right)}
						d \R\\
&=
h\left(\widehat{\R}\right)+\frac{1}{2} \log \left|2 \pi \mathbb{\H}\right|						\\
&= \log p\left(\widehat{\r}\right)
+\log p\left(\y|\widehat{\r}\right) + \left|2 \pi \mathbb{\H}^{-1}\right|^{1/2}.
\end{split}						
\end{align}  
Finally, the marginal likelihood estimate can be written as
\begin{align}
\label{laplace_est}
\widehat{p}\left(\y\right)=
p\left(\widehat{\r}\right)
p\left(\y |\widehat{\r}\right)
\left|2 \pi \mathbb{\H}^{-1}\right|^{1/2}.
\end{align}  
The Laplace approximation to the marginal likelihood consists of a term for the data likelihood at the mode (second term of (\ref{laplace_est})), a penalty term from the prior (first term of (\ref{laplace_est})), and a volume term calculated from the local curvature (third term of (\ref{laplace_est})).\\
Under the Laplace approximation presented in (\ref{laplace_est}), the LRT decision rule in (\ref{LRT}) is approximated by
\begin{equation}
\widehat{\Lambda}\left(\Y_{1:L}\right) =
\frac{\prod_{l=1}^L \widehat{p}\left(\y(l)| \Halt\right)}
{\prod_{l=1}^L \widehat{p}\left(\y(l)| \Hnull\right) },
\begin{array}{c}
\stackrel{\Hnull}{\geq} \\
\stackrel{<}{\Halt}
\end{array}%
\gamma 
\end{equation}
where $\widehat{p}\left(\y(l)| \mathcal{H}_k\right)$ is the Laplace marginal likelihood approximation under the $k$-th hypothesis.
The major difficulty in evaluating (\ref{laplace_est}) is the requirement to evaluate the MAP estimate $\widehat{\R}$ under each hypotheses. This task is nontrivial as it involves a non-convex and non-linear optimisation problem. We derive the MAP estimate for this scenario via the Bayesian Expectation Maximasation (BEM) methodology, see the derivation in the Appendix .
%
%
%

\section{Simulation Results} \label{SimulationResults}
In this section, we present the performance of the proposed algorithms via Monte Carlo simulations.
\subsection{Simulation Set-up}
The simulation settings for all the simulations are as follows:
\begin{itemize}
	\item The prior distribution for all the channels is Rayleigh fading, and the channels are assumed to be both spatially and temporally independent.
\item We define the \textit{receive SNR} as the ratio of the average received signal power to the average noise power
\begin{align*}
\text{SNR} \triangleq 
10 \log \frac
{\text{Tr}\left[\exE\left[\left(\G(l) \F(l) s(l)\right)\left(\G(l) \F(l) s(l)\right)^H\right]\right]}
{\text{Tr}\left[\exE\left[\left(\G(l) \V(l) + \W(l)\right)\left(\G(l) \V(l) +\W(l)\right)^H\right]\right]}=
10 \log \frac
{1}
{\sigma^2_{\V}+\frac{1}{M}\sigma^2_{\W}}.
\end{align*}
\item The SNR is set to $0$ dB.
\item The results are obtained from simulations over $100,000$ channels and noise realisations for a given set of $N$, $M$ and $L$.
\item For the Laguerre series expansion, the order of the series expansion was set to $p=100$.
\end{itemize}
\subsection{Study of detection probability Vs. frame length}
In this section we study the relationship between the ability to detect the presence of a signal in a spectrum sensing problem as a function of the length of the frame, $L$. We undertake this study in two different scenarios, the first involves perfect CSI according to Section \ref{PERFECT_CSI} and the second involves partial CSI according to Section \ref{Laguerre_section}.
We set the channels uncertainty $\overline{\F}(l)=\bm{0}$ and $\sigma^2_{\F} = 1$, thus only prior information is available for the $\F$ channels.

In presenting results we fix the false alarm rate $p_f$ to $10\%$.
We repeat this study for a range of values of the number of receive antennas, $N \in \left\{1, 2, 4, 8\right\}$. The results are depicted in Fig. \ref{fig:fig11} and they demonstrate the following key points: 
\begin{enumerate}
	\item for all frame lengths, as the number of receive antennas is increased, the probability of detection improves as expected;
	\item for all frame length, the detection probability under perfect CSI always outperforms significantly the performance of the model with partial CSI;
	\item asymptotically in the frame length, $L$, the probability of detection for any number of receive antennas converges to $1$, with different rates, depending on $N$;
	\item such a study provides generic performance specifications that allow us to obtain the same detection probability for different combinations of frame length and number of receive antennas. For example, with $L=10$ and $N=1$, this will be equivalent to $L=3$ and $N=4$.
	\item it also guides system design that for a given desired probability of detection, we see the saturation point, after which, increasing the frame length delivers negligible improvement. 
\end{enumerate}
\textbf{Evaluation of the Gaussian approximation}\\
The accuracy of the GA in Section \ref{Gaussian_Approximation} is bounded via a multi-dimensional Berry-Esseen inequality in Theorem \ref{theorem_Berry_Esseen}. Here we study this accuracy using a graphical Q-Q plots of each element of $\Y(l)$ as a function of the number of relays $M$. 
The results are presented in Fig. \ref{fig:fig10} and demonstrate that for a fixed frame length and number of receive antennas, as one increases the number of relays $M$, the Gaussian approximation that we made in Lemma \ref{Gaussian_apprx} improves. We see that in the setting of partial CSI which is relevant to practical scenarios, the number of relays required before one can make a reasonable Gaussian approximation is around $8$.

\subsection{Comparison of Detection Probability under different LRT Statistic Approximations}
In this section we present a comprehensive comparison of the distributional estimators derived for the LRT test statistic in order to evaluate the probability of detection.
This is undertaken in a range of different scenarios and we compare the distributional estimates under different levels of CSI versus the best case scenario bounds. The comparison is undertaken between:
\begin{enumerate}
	\item the analytic evaluations of the probabilities of detection and false alarm under the setting of perfect CSI, according to results obtained in Theorem \ref{Theorem_1} (denoted by: CSI Theory);
	\item  the Monte Carlo based empirical estimation of the probabilities of detection and false alarm under the setting of perfect PBS-relays CSI, and perfect relays-SBS CSI, according to the decision rule derived in Theorem \ref{Theorem_3} (denoted by: CSI empirical);	
	\item the analytic evaluations of the probabilities of detection and false alarm under the setting of imperfect PBS-relays CSI, and perfect relays-SBS CSI, according to the Laguerre series expansion density approximations derived in Theorem \ref{theorem_Laguerre} in (\ref{p_d_csi_laggure}-\ref{p_f_csi_laggure}) (denoted by: P-CSI Laguerre);
	\item the Monte Carlo based empirical estimation of the probabilities of detection and false alarm under the setting of imperfect PBS-relays CSI, and perfect relays-SBS CSI, according to the decision rule derived in Theorem \ref{Theorem_9} (denoted by: P-CSI empirical);

	\item the Monte Carlo based Gaussian approximation of the probabilities of detection and false alarm under the setting of imperfect PBS-relays CSI, and imperfect relays-SBS CSI, according to Lemma \ref{Gaussian_apprx} applied to the decision rule derived in Theorem \ref{Theorem_9} (denoted by: PP-CSI Gaussian);

	\item the Monte Carlo based Laplace approximation of the probabilities of detection and false alarm under the setting of imperfect PBS-relays CSI, and imperfect relays-SBS CSI, corresponding decision rule also derived (denoted by: PP-CSI Laplace).

\end{enumerate}
The scenarios we consider involve varying the number of receive antennas $N$ and the number of relays $M$, for a fixed frame length $L=1$ and a fixed SNR of $0$ dB. The Receiver Operating Characteristic (ROC) curves are presented in Figs. \ref{fig:fig1}- \ref{fig:fig6}, for each of these comparisons.
The following summary details the key points of this analysis:
\begin{enumerate}
	\item In all study combinations of $N$ and $M$, the probability of detection for each probability of false alarm, had an ordering of algorithmic performance, in agreement with theory, given by:
\begin{enumerate}	
\item [i.] Optimal performance under perfect CSI. This results in the theoretical upper bound of Theorem \ref{Theorem_1} which agreed exactly with the Monte Carlo estimate under this scenario.
\item [ii.] This was followed by the results of the imperfect PBS-relays CSI, and perfect relays-SBS CSI which were obtained under the Laguerre approximation and again compared to a Monte Carlo simulation estimated.
\item [iii.] Finally the results of the approximations when least information is known, imperfect PBS-relays CSI, and imperfect relays-SBS CSI which were obtained under the Laplace approximation and the Gaussian approximation. The Laplace approximation outperformed the Gaussian approximation in situations in which the distribution of the test statistic was not close to Gaussian. 
\end{enumerate}	  
\item In all the examples the Laplace approximation outperformed the Gaussian approximation or was directly comparable in performance as the Central Limit Theorem became viable, i.e. when $M$ was large, as presented in Fig.\ref{fig:fig3}.
\end{enumerate}
\section{Conclusions and Future Work} \label{Conclusions}
In this paper we developed a framework for spectrum sensing in cooperative amplify-and-forward cognitive radio networks.
We developed the Bayesian optimal decision rule under various scenarios of CSI varying from perfect to imperfect CSI.
We designed two algorithms to approximate the marginal likelihood, and obtained the decision rule. We utilised a Laguerre series expansion to approximate the distribution of the test statistic in cases where its distribution can not be derived exactly.
Future research will include comparison of the Laplace method to other low complexity approaches, such as the Akaike and Bayesian information
criteria.
\section{Acknowledgment}
The authors would like to thank Professor Antonia Casta\~no-Mart\'inez from the University of C\'adiz, Spain for the valuable comments and for providing the maple code of the Laguerre series.

\bibliographystyle{IEEEtran}
\bibliography{../../../references}
\newpage
\appendix 
\textbf{Proof of Theorem \ref{theorem_normal_product}}
\begin{proof}
Consider the normalised product of two independent normally distributed random variables defined as
\begin{equation}
Z=\frac{X Y}{\sigma_X \sigma_Y},
\end{equation}
 where $X\sim N\left(\overline{X}, \sigma^2_X\right)$ and $Y\sim N\left(\overline{Y}, \sigma^2_Y\right)$, and define $\rho_X = \frac{\overline{X}}{\sigma_X }$ and $\rho_Y = \frac{\overline{Y}}{\sigma_Y }$. \\
The expectation of $Z$ is given by
\begin{equation*}
\overline{Z} = \exE\left[Z\right]=\frac{\exE\left[X\right] \exE\left[Y\right]}{\sigma_X \sigma_Y} = \rho_X \rho_Y,
\end{equation*}
The variance of $Z$ is given by
\begin{equation*}
\sigma^2_Z = \exE\left[Z^2\right]-\exE\left[Z\right]^2=
\frac{\exE\left[X^2\right] \exE\left[Y^2\right]}
{\sigma^2_X \sigma^2_Y}-\overline{Z}^2=
\frac{ 
\left(\sigma^2_X+\overline{X}^2\right) 
\left(\sigma^2_Y+\overline{Y}^2\right)
}
{\sigma^2_X \sigma^2_Y} -\exE\left[Z\right]^2=
1+\rho_X^2+\rho_Y^2.
\end{equation*}
We define $\widetilde{Z} = \left(Z-\overline{Z}\right)/\sigma_Z$ and derive the MGF of $\widetilde{Z}$ using (\ref{MGF}):
\begin{equation}
\begin{split}
M_{\widetilde{Z}} \left(t\right)  =& 
M_{\left(Z-\overline{Z}\right)/\sigma_Z} \left(t\right)  = 
\exE\left[\exp \left\{ 
\frac{\left(Z-\overline{Z}\right)}{\sigma_Z} t 
\right\}\right]\\
=&\exp \left\{ 
-\frac{\overline{Z}}{\sigma_Z} t 
\right\}
\exE\left[
\exp \left\{ 
\frac{Z}{\sigma_Z} t 
\right\}
\right]\\
=&
\exp \left\{ 
-\frac{\overline{Z}}{\sigma_Z} t 
\right\}
M_{Z} \left(\frac{t}{\sigma_Z}\right)\\
=&
\frac{\exp \left\{ 
-\frac{\rho_X \rho_Y}{\sqrt{1+\rho_X^2+\rho_Y^2}} t 
\right\}
\exp\left\{
\frac{\frac{\left(\rho_X^2+\rho_Y^2\right)t^2}{1+\rho_X^2+\rho_Y^2}+ 
\frac{2 \rho_X \rho_Y t}{\sqrt{1+\rho_X^2+\rho_Y^2}}}
{2\left(1-\frac{t^2}{1+\rho_X^2+\rho_Y^2}\right)}\right\}}
{\sqrt{1-\frac{t^2}{1+\rho_X^2+\rho_Y^2}}}\\
=&
\frac{
\exp\left\{
\frac{\frac{\left(\rho_X^2+\rho_Y^2\right) t^2}{1+\rho_X^2+\rho_Y^2}
+\frac{2 \rho_X \rho_Y t}{\sqrt{1+\rho_X^2+\rho_Y^2}}}
{2\left(1-\frac{t^2}{1+\rho_X^2+\rho_Y^2}\right)}
-\frac{\rho_X \rho_Y t}{\sqrt{1+\rho_X^2+\rho_Y^2}} \right\}}
{\sqrt{1-\frac{t^2}{1+\rho_X^2+\rho_Y^2}}}\\
=&
\frac{
\exp\left\{
\frac{\frac{\left(\rho_X^2+\rho_Y^2\right) t^2}{1+\rho_X^2+\rho_Y^2}
+\frac{2 \rho_X \rho_Y t}{\sqrt{1+\rho_X^2+\rho_Y^2}} 
-\frac{2\rho_X \rho_Y t}{\sqrt{1+\rho_X^2+\rho_Y^2}}
\left(1-\frac{t^2}{1+\rho_X^2+\rho_Y^2}\right)
}
{2\left(1-\frac{t^2}{1+\rho_X^2+\rho_Y^2}\right)}
 \right\}}
{\sqrt{1-\frac{t^2}{1+\rho_X^2+\rho_Y^2}}}\\
=&
\frac{
\exp\left\{
\frac{\frac{\left(\rho_X^2+\rho_Y^2\right) t^2}{1+\rho_X^2+\rho_Y^2}
+\frac{2\rho_X \rho_Y t}{\sqrt{1+\rho_X^2+\rho_Y^2}}
\frac{t^2}{1+\rho_X^2+\rho_Y^2}
}
{2\left(1-\frac{t^2}{1+\rho_X^2+\rho_Y^2}\right)}
 \right\}}
{\sqrt{1-\frac{t^2}{1+\rho_X^2+\rho_Y^2}}}.
\end{split}
\end{equation}
Finally, we define $\alpha = \frac{t}{\sqrt{1+\rho_X^2+\rho_Y^2}}$ and take the limits $\rho_X  \to \infty$ or $\rho_Y  \to \infty$  to obtain the following standard normal distribution:
\begin{equation}
\lim_{
\begin{array}{c}
\rho_X  \to \infty\\
\rho_Y  \to \infty
\end{array}
 }
 M_{\widetilde{Z}}=
 \lim_{
 \alpha  \to 0
 }
\frac{
\exp\left\{
\frac{\left(\rho_X^2+\rho_Y^2\right) \alpha^2
+2\rho_X \rho_Y \alpha^3
}
{2\left(1-\alpha^2\right)} \right\}}
{\sqrt{1-\alpha^2}}
= \exp\left\{\frac{t^2}{2}\right\},
\end{equation}
which is the MGF of $N\left(0,1\right)$.
Therefore, we obtain that $Z\sim \left(\overline{Z},\sigma_Z^2\right)$.
\end{proof}
\textbf{Deriving the MAP estimate of $\widehat{\R}$ in (\ref{laplace_est})}\label{Appendix_1}\\
The MAP optimisation problem can be written as
\begin{align}  
\label{MAP_estimation1}
\widehat{\R} = \argmax_{\r} p\left(\r |\y \right)
=\argmax_{\r} p\left(\y |\r \right)p\left(\r \right),
\end{align} 
where $p\left(\y |\r \right)$ and $p\left(\r \right)$ are defined in (\ref{distributions_1})-(\ref{distributions_2}). Then, the MAP estimate is the solution for the following optimisation problem, where for simplicity we remove the time dependence $l$:
\begin{align}  
\label{MAP_estimation}
\begin{split}  
\widehat{\R} &= \argmax_{\r} p\left(\y|\r\right)p\left(\r\right)\\
&= \argmax_{\r} \underbrace{\frac{1}{\left(\sigma^{2}_{\G} \left\|\r\right\|^{2}+\sigma^{2}_{\W}\right)^{N}}
		\exp^{\left(-\frac{\left\|\y- \overline{\G}\r\right\|^{2}}{\sigma^{2}_{\G}\left\|\r\right\|^{2} +\sigma^{2}_{\W}}\right)}}_{p\left(\y|\r\right)}
\times 
\underbrace{\frac{1}{\left(\sigma^{2}_{\R}\right)^{M}} \exp^{\left(-\frac{\left\|\r-\overline{\R}\right\|^{2}}   {\sigma^{2}_{\R}}\right)}}_{p\left(\r\right)},
\end{split}  
\end{align}  
where
\begin{align}
\begin {cases}
\begin{split}
    &\Hnull: \overline{\R} = \mathbf{0}, \hspace{0.25cm}\sigma^{2}_{\R} = \sigma^2_{\V}\\
    &\Halt: \overline{\R} = \overline{\F}s, \sigma^{2}_{\R} =  \sigma^2_{\F} +\sigma^2_{\V}.
\end{split}
\end {cases} 
\end{align}

Problem (\ref{MAP_estimation}) is non-linear and non-convex. We shall utilise the Bayesian Expectation Maximisation (BEM) methodology to solve it efficiently under each hypothesis.
The BEM algorithm (see \cite{dempster:1977}, \cite{baum:1970}) is an iterative method that alternates between an E step, which infers posterior distributions over hidden variables given a current parameter setting, and an M step, which maximises $p\left(  \y  ,\G ,\R\right) $ with respect to $\R$ given the
statistics gathered from the E step. 
The BEM can be easily evaluated using the following iterative steps, at iteration $(n+1)$:
\begin{subequations}
\begin{align}
\label{BEM_algorithm_E}
&\text{\textbf{E Step}:\;\;} L\left(\R\right) = 
\exE_{\G |\y;\widehat{\R}^n } \left[\log p\left(  \y  ,\G ,\R\right) \right]\\
\label{BEM_algorithm_M}
&\text{\textbf{M Step}:\;\;} \widehat{\R}^{n+1} = \argmax_{\R}L\left(\R\right)
\end{align}
\end{subequations}
The \textbf{E Step} can be expressed as:
\begin{align}
\label{E_step}
\begin{split}
L(\R) &= \exE_{\G |\Y;\widehat{\R}^n } \left[\log p\left(  \Y  ,\G ,\R\right)   \right]\\
		  &= \exE_{\G |\Y;\widehat{\R}^n } \left[\log p\left(  \Y  |\G ,\R\right) +\log p\left(\R|\cancel{\G}\right)  \right]\\
      &= \exE_{\G |\Y ;\widehat{\R}^n }\left[-\frac{1}{\sigma^{2}_{\W}}\left\|\Y -\G  \R \right\|^{2}-
        	\frac{1}{\sigma^{2}_{\R}}\left\| \R-\overline{\R}\right\|^{2} \right]+\text{constant}\\
      &= \frac{1}{\sigma^{2}_{\W}}
      \left(2 \Y^H \R \exE_{\G |\Y ;\widehat{\R}^n }\left[\G\right]  -
      \R^T \exE_{\G |\Y ;\widehat{\R}^n }\left[\G^T\G\right] \R \right)- 
      \frac{1}{\sigma^{2}_{\R}}\left( \R^H \R- 2\R^H\overline{\R} \right)+\text{constant}.
\end{split}
\end{align}
where constant contains all terms that are independent of $\R$.

The conditional expectations in (\ref{E_step})can be evaluated using Bayesian MMSE as follows:
we first re-write the observation model (\ref{observation_model}) as:
\begin{align}
\begin{split}
\Y &=\G\R + \W= \left(\R^T  \otimes \I\right) \text{vec}\left[ \G\right] + \W = \Omega \GAMMA + \W ,
\end{split}
\end{align}
where we define $\Omega \triangleq \left(\R^T  \otimes \I\right) $, $\GAMMA \triangleq \text{vec}\left[ \G\right]$, and  $\otimes$ is the Kronecker product, and $\text{vec}\left[\cdot\right]$ is the vector obtained by stacking the columns of a matrix one over the other.
Since $\Y$ and $\GAMMA$ are jointly Gaussian, the Linear MMSE is also the MMSE estimator (see \cite{kay:1998}). The LMMSE can be expressed as

\begin{align}
\begin{split}
\label{MMSE_expectations_1}
\exE_{\GAMMA |\Y ;\widehat{\R}^n }\left[\GAMMA \right]
&=\exE\left[\GAMMA\right]+ \exE\left[\GAMMA \Y^H\right] \exE^{-1}\left[\Y \Y^H\right]\left(\Y-\exE\left[\Y\right] \right)\\
&=\overline{\GAMMA}+ \exE\left[\GAMMA \Y^H\right] \exE^{-1}\left[\Y \Y^H\right]\left(\Y-\Omega\overline{\GAMMA} \right)\\
&=\overline{\GAMMA}+   \frac{\Omega^{H} \left(\Y-\Omega\overline{\GAMMA} \right)}
									{\left\|\widehat{\R}^n \right\|^{2}+ \frac{\sigma^{2}_{W}}{\sigma^{2}_{G}}},
\end{split}
\end{align}
where $\overline{\GAMMA} = \exE\left[\GAMMA\right]$.
Next, we evaluate the covariance matrix:
\begin{align}
\begin{split}
\label{MMSE_expectations_2}
\text{Cov}_{\GAMMA |\Y ;\widehat{\R}^n }\left[\GAMMA \right] &=
\exE\left[\GAMMA \; \GAMMA^H\right]-
\exE\left[\GAMMA \;\Y^H\right]
\exE^{-1}\left[\Y \;\Y^H\right]
\exE\left[\Y \;\GAMMA^H\right]\\
&=
\sigma^{2}_{G} \I -\frac{\sigma^{2}_{G}\left(\widehat{\R}^n  \left(\widehat{\R}^n \right)^T  \otimes \I\right)}
												{	\left\|\widehat{\R}^n \right\|^2+	\frac{\sigma^{2}_{W}}{\sigma^{2}_{G}}} .
\end{split}
\end{align}
By rearranging the above expressions, we obtain
\begin{subequations}
\begin{align}
\label{MMSE_expectations_1}
\Phi_{1}\left(\Y ,\widehat{\R}^n \right) &\triangleq \exE_{\G |\Y ;\widehat{\R}^n }\left[\G \right]
= \overline{\G}+   \frac{1}{\left\|\widehat{\R}^n \right\|^{2}+ \frac{\sigma^{2}_{W}}{\sigma^{2}_{G}}}\left(\Y-\overline{\G} \widehat{\R}^n \right)\left(\widehat{\R}^n \right)^{H},\\
\label{MMSE_expectations_2}
\Phi_{2}\left(\Y ,\widehat{\R}^n \right) &\triangleq \exE_{\G |\Y ;\widehat{\R}^n }\left[\G^{H} \G \right] 
=
    \Phi_{1}\left(\Y,\widehat{\R}^n \right)^{T}\Phi_{1}\left(\Y,\widehat{\R}^n \right)
    +  \sigma^{2}_{G} N\left(\I-\frac{\sigma^{2}_{G} \widehat{\R}^n  \left(\widehat{\R}^n \right)^{T}} {\left\|\left(\widehat{\R}^n \right)\right\|^{2}+
    \frac{\sigma^{2}_{W}}{\sigma^{2}_{G}}}
     \right).
\end{align}
\end{subequations}
Using (\ref{MMSE_expectations_1}-\ref{MMSE_expectations_2}), (\ref{E_step}) can be expressed as
\begin{align}
L(\R) = \frac{1}{\sigma^{2}_{\W}}
      \left(2 \Y^H \R \Phi_{1}\left(\Y ,\widehat{\R}^n \right)   -
      \R^T \Phi_{2}\left(\Y ,\widehat{\R}^n \right)  \R \right)- 
      \frac{1}{\sigma^{2}_{\R}}\left( \R^H \R- 2\R^H\overline{\R} \right)+\text{constant}.
\end{align}
The \textbf{M Step} is obtained by setting to $0$ the derivative of $L(\R)$ with respect to $\R$:
\begin{align}
\begin{split}
\widehat{\R}^{n+1} = \argmax_{\R} L(\R)
        	=\left(\Phi_{2}\left(\Y,\R^{k}\right)+\frac{\sigma^{2}_{\W}}{\sigma^{2}_{\R}}\I\right)^{-1}
\left(\Phi_{1}\left(\Y,\R^{k}\right)^{T}\Y
+ \overline{\R} \frac{\sigma^{2}_{\W}}{\sigma^{2}_{\R}}\right).
\end{split}
\end{align}
The BEM algorithm requires that $\widehat{\R}^{n+1}$ is initialised at $n=0$. The simplest option is to initialise it to the prior, that is $\widehat{\R}^{0} = \overline{\R}.$
\begin{table} 	[h]
  \begin{center}
    	\begin{tabular}{|c|c|c|c|l|l|}
			\addlinespace[1pt	]
			\multicolumn{5}{c}{} \\			\hline
			Case &PBS-relays ($\F$)&relays - SBS ($\G$) &  Section & Decision rule & Performance analysis
			\\
			\hline\addlinespace[2pt	]
			I&$\surd$&$\surd$& \ref{PERFECT_CSI}& Exact 		&Exact analytic\\
			\hline\addlinespace[1pt	]
			II&$\times$&$\surd$&\ref{Laguerre_section}& Exact & Analytic approximation  		\\
			\addlinespace[1pt	]
			$$&$$&$$&$$& $$ & via Generalized Laguerre polynomials 		\\
			\hline\addlinespace[1pt	]
			III&$\surd$&$\times$&\ref{App_Algorithms}&  Special case of IV	& see Section IV\\
			\hline\addlinespace[1pt	]
			IV&$\times$&$\times$&\ref{App_Algorithms}& Analytic approximation	& Simulation \\
			\addlinespace[1pt	]
			$$&$$&$$&$$& via Laplace integrals	& $$ \\
			\addlinespace[1pt	]
			IV&$\times$&$\times$&\ref{App_Algorithms}& Analytic approximation  & Simulation \\
			\addlinespace[1pt	]
			$$&$$&$$&$$& via Moments	matching  & $$ \\
			\addlinespace[1pt	]
			\hline
			\end{tabular}
\end{center}
\caption{Summary of proposed solutions based on CSI knowledge} 
\label{Table_algorithms}
\end{table}     

\begin{figure}[b]
    \centering
        \epsfysize=5cm
        \epsfxsize=12cm
        \epsffile{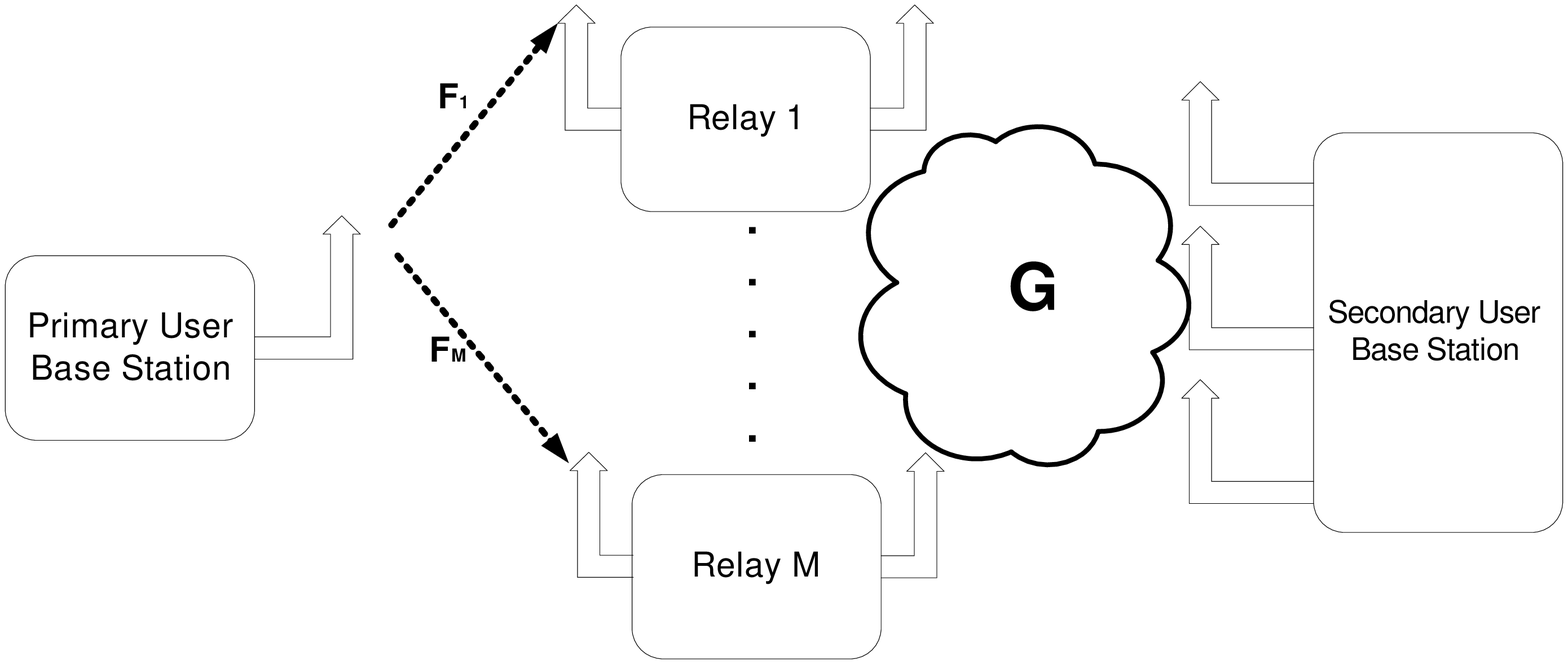}
        \caption{System model of Cooperative Cognitive Radio network with $M$ relays and a multiple antenna receiver}
    \label{fig:system_model}
\end{figure}

\begin{figure}
  \begin{center}
    \centering
    \epsfysize=9cm
    \epsfxsize=10cm
    \epsffile{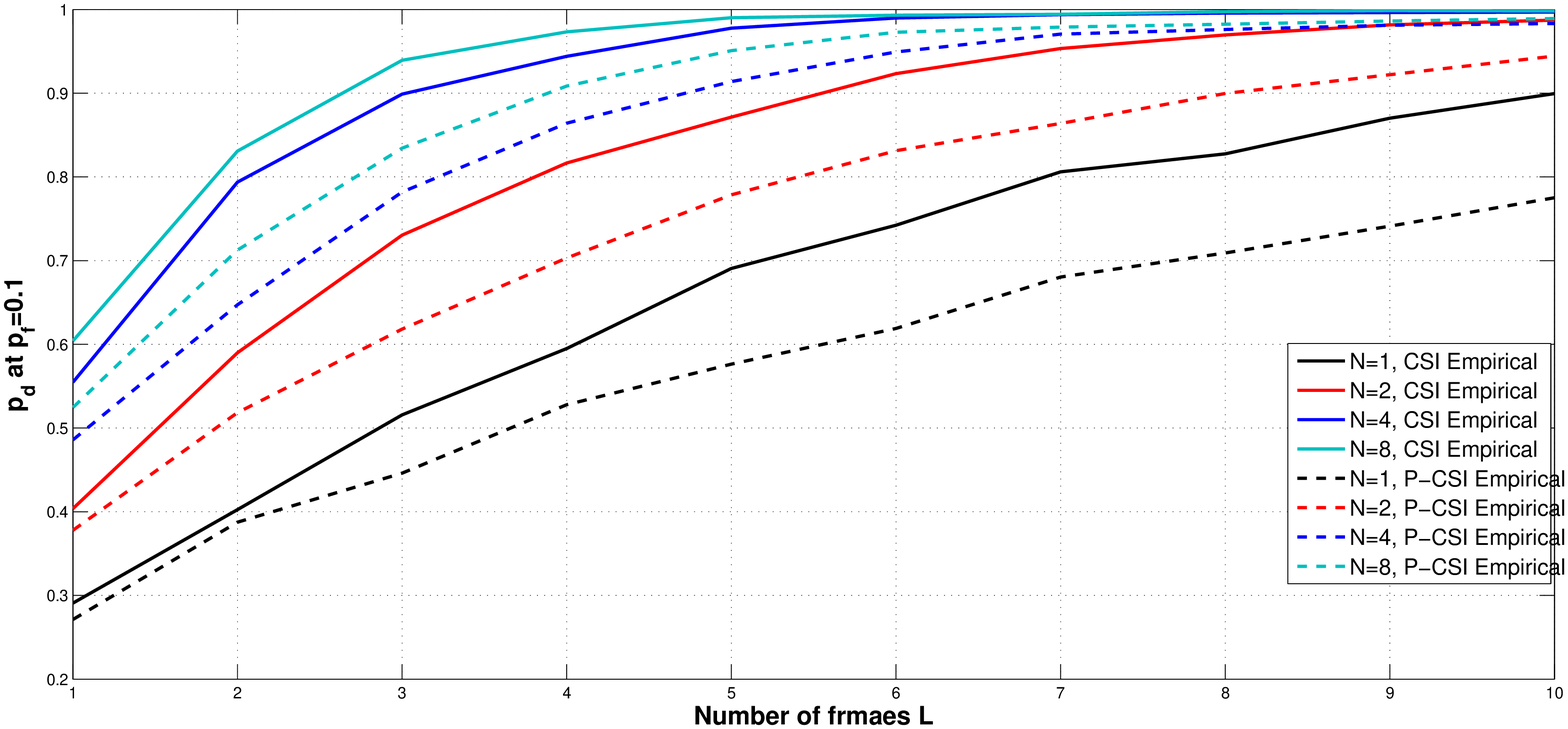}
    \caption{Probability of detection at $p_f=0.1$ for the cases of perfect CSI (Section \ref{PERFECT_CSI}) and imperfect CSI (Section \ref{Laguerre_section})} 
    \label{fig:fig11}
  \end{center}
\end{figure}

\begin{figure}
  \begin{center}
    \centering
    \epsfysize=9cm
    \epsfxsize=10cm
    \epsffile{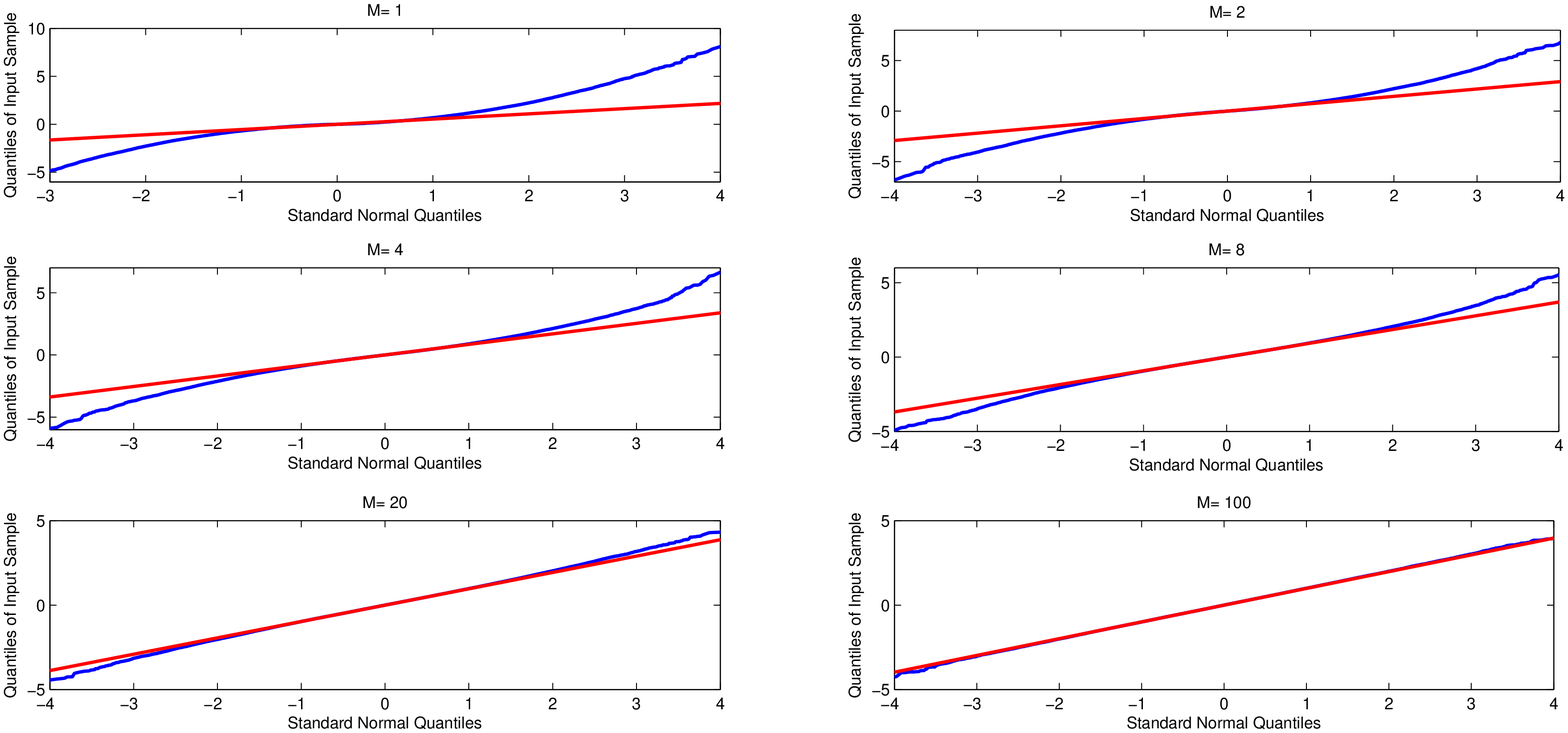}
    \caption{Q-Q plot of the normal approximation per Section \ref{Gaussian_Approximation} for different number of relays}
    \label{fig:fig10}
  \end{center}
\end{figure}

\begin{figure}
  \begin{center}
        \centering
	\epsfysize=9cm
    \epsfxsize=12.5cmcm
    \epsffile{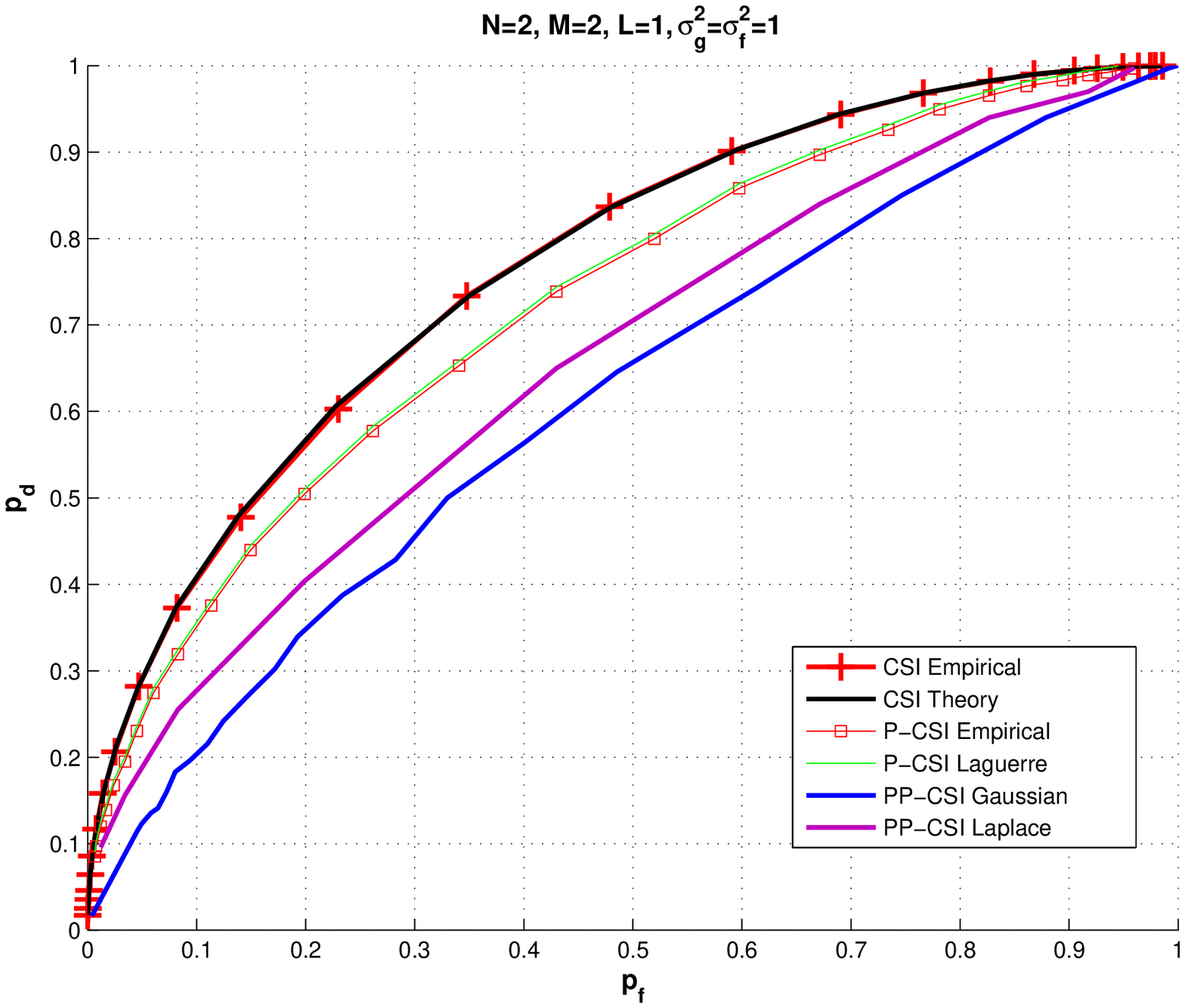}
    \caption{Probability of detection vs. probability of false alarm for $N=2$, $M=2$, $L=1$} 
    \label{fig:fig1}
  \end{center}
\end{figure}

\begin{figure}
  \begin{center}
        \centering
    \epsfysize=9cm
    \epsfxsize=12.5cmcm
    \epsffile{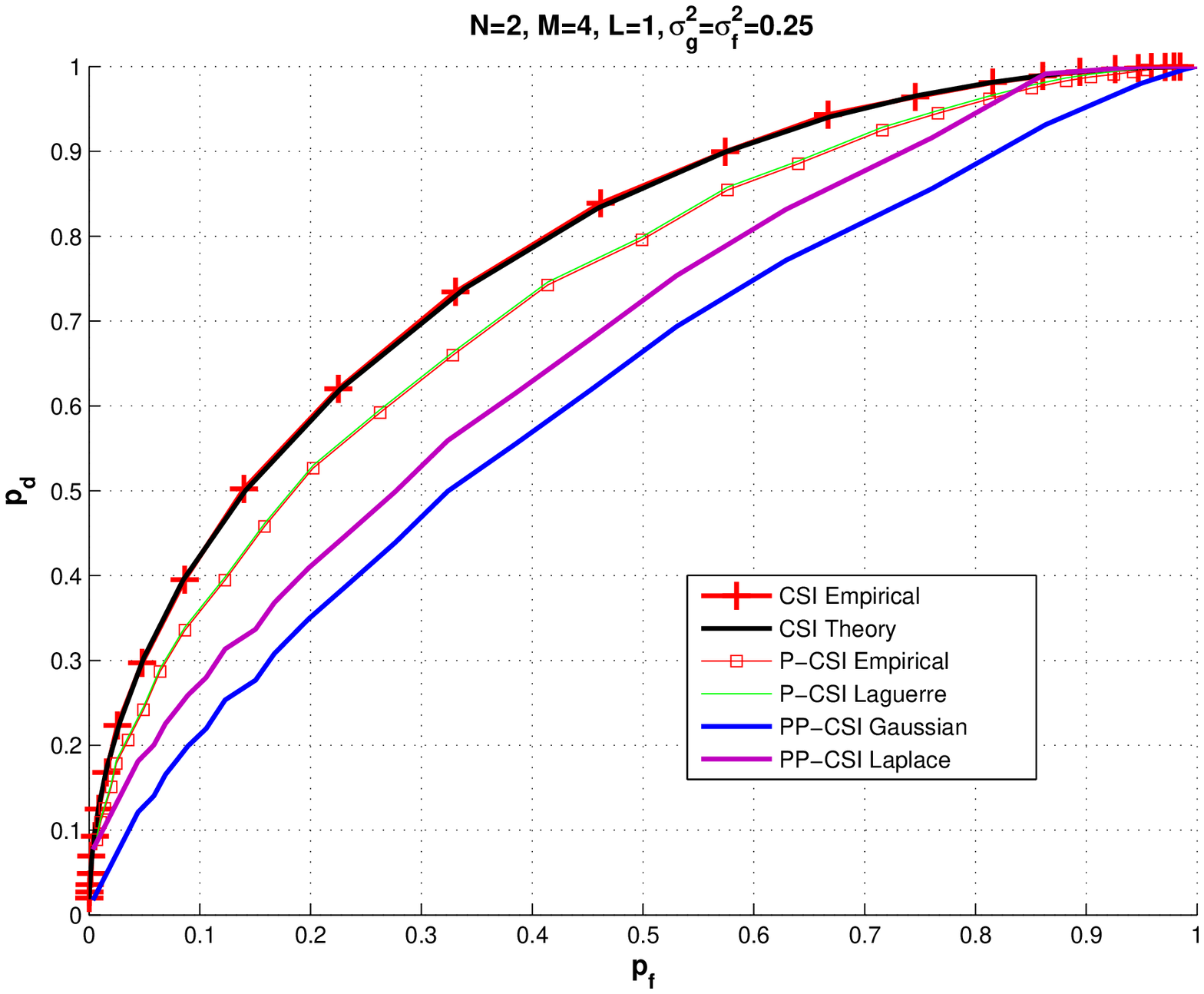}
\caption{Probability of detection vs. probability of false alarm for $N=2$, $M=4$, $L=1$}
    \label{fig:fig2}
  \end{center}
\end{figure}

\begin{figure}
  \begin{center}
   \centering
    \epsfysize=9cm
    \epsfxsize=12.5cm
    \epsffile{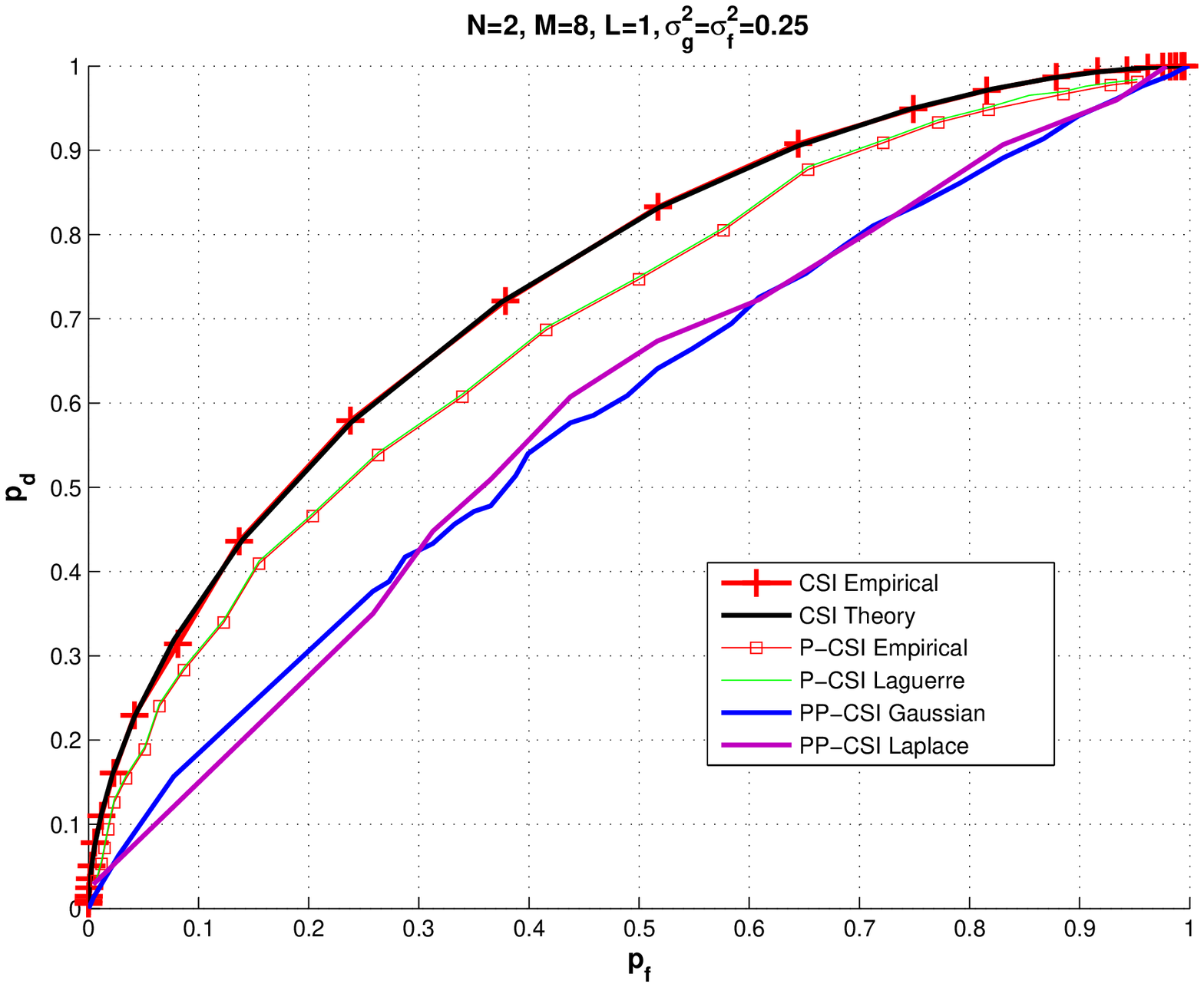}
    \caption{Probability of detection vs. probability of false alarm for $N=2$, $M=8$, $L=1$}
    \label{fig:fig3}   
  \end{center}
\end{figure}

\begin{figure}
  \begin{center}
   \centering
    \epsfysize=9cm
    \epsfxsize=12.5cm
    \epsffile{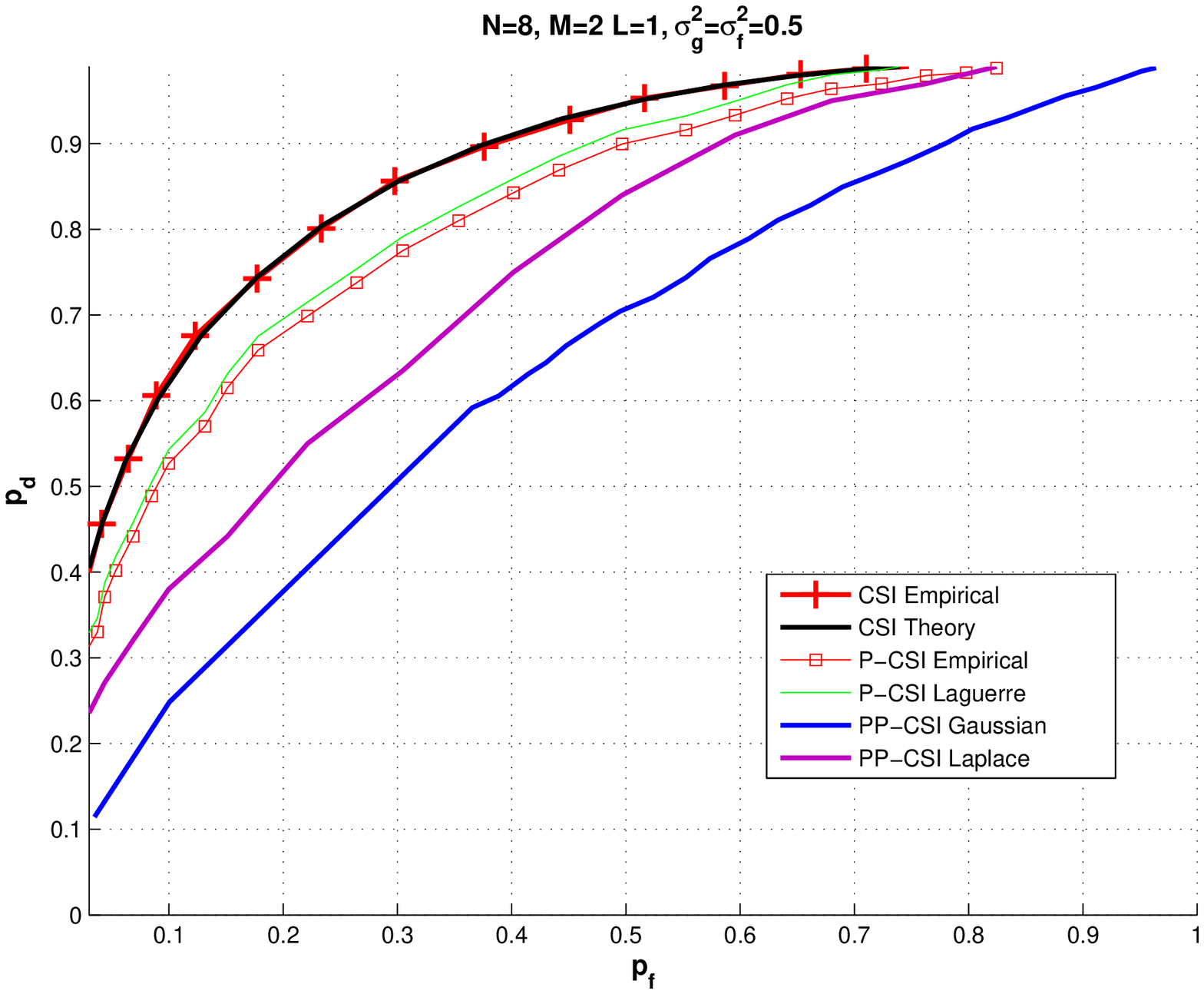}
    \caption{Probability of detection vs. probability of false alarm for $N=8$, $M=2$, $L=1$} 
    \label{fig:fig4}    
  \end{center}
\end{figure}

\begin{figure}
  \begin{center}
  \centering
    \epsfysize=9cm
    \epsfxsize=12.5cm
    \epsffile{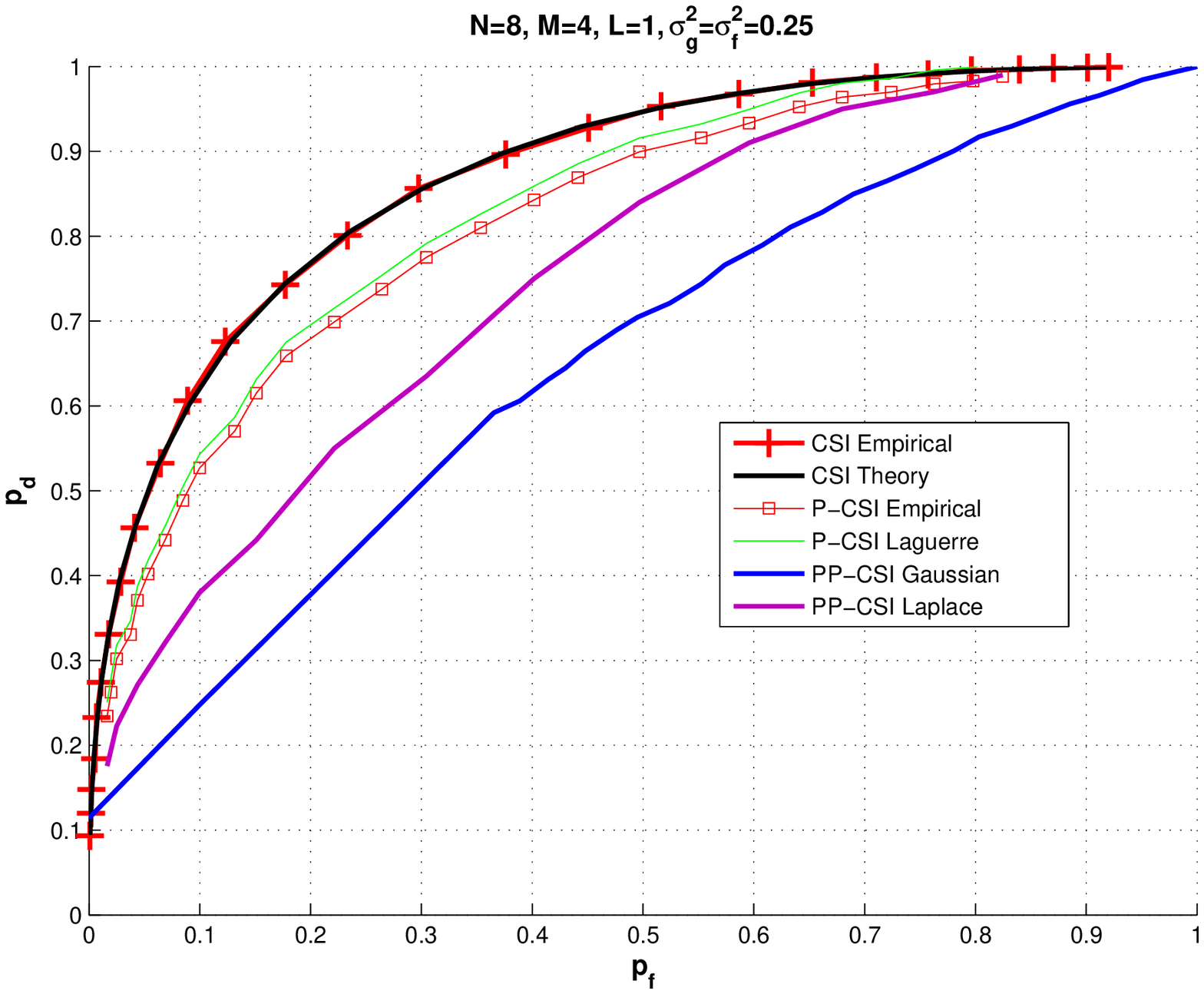}
    \caption{Probability of detection vs. probability of false alarm for $N=8$, $M=4$, $L=1$}
    \label{fig:fig5}   
  \end{center}
\end{figure}

\begin{figure}
  \begin{center}
    \centering
    \epsfysize=9cm
    \epsfxsize=12.5cm
    \epsffile{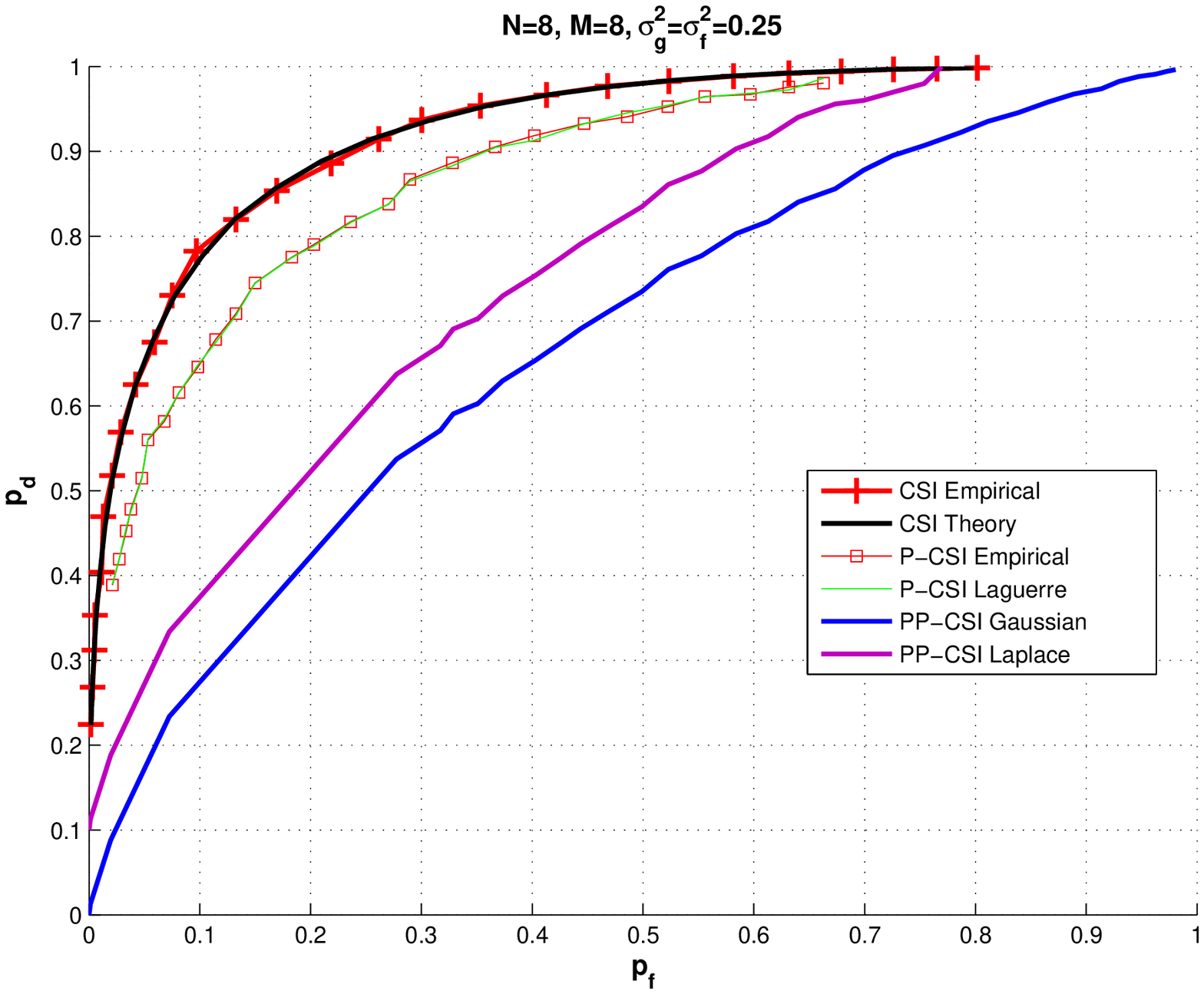}
    \caption{Probability of detection vs. probability of false alarm for $N=8$, $M=8$, $L=1$} 
    \label{fig:fig6}
  \end{center}
\end{figure}
\end{document}